\documentclass[english]{amsart}
\usepackage{babel}
\usepackage[]{geometry}

\usepackage{caption}
\usepackage{subcaption}

\usepackage{amsmath,amssymb,amsfonts,latexsym,cancel}
\usepackage{rawfonts}
\usepackage{pictexwd}
\usepackage{tikz}

\usepackage{tikz-cd}

\usetikzlibrary{positioning}

\usepackage{pgfplots}

\pgfplotsset{compat=newest}
\pgfplotsset{plot coordinates/math parser=false}
\newlength\figureheight
\newlength\figurewidth

\usepackage{color}
\usepackage{epsfig}

\newcommand{\R}{\mathbb{R}}
\newcommand{\N}{\mathbb{N}}

\def\1{\raisebox{2pt}{\rm{$\chi$}}}

\theoremstyle{plain}

\newtheorem{lemma}{Lemma}[section]
\newtheorem{remark}{Remark}[section]
\newtheorem{assumption}{Assumption}

\theoremstyle{definition}

\theoremstyle{remark}

\numberwithin{equation}{section}

\usepackage{hyperref}

\usepackage[foot]{amsaddr}

\begin{document}

\title{Spectral decomposition of atomic structures in heterogeneous cryo-EM}

\author{Carlos Esteve-Yag\"ue}
\email{ce423@cam.ac.uk}
\author{Willem Diepeveen}
\email{wd292@cam.ac.uk}
\address[Carlos Esteve-Yag\"ue, Willem Diepeveen and Carola-Bibiane Sch\"onlieb]{Department of Applied Mathematics and Theoretical Physics, University of Cambridge, United Kingdom.}

\author{Ozan \"Oktem}
\address[Ozan \"Oktem]{Department of Mathematics, KTH, Stockholm, Sweden.}
\email{ozan@kth.se}

\author{Carola-Bibiane Sch\"onlieb}
\email{cbs31@cam.ac.uk}

\thanks{\textbf{Acknowledgements:} The authors acknowledge support from the Alan Turing Institute for the project "Molecular structure from images under physical constraints”. O.  \"Oktem was partially supported by the Swedish Research Council Grant 2020-03107. C.B. Sch\"onlieb acknowledges support from the Philip Leverhulme Prize, the Royal Society Wolfson Fellowship, the EPSRC advanced career fellowship EP/V029428/1, EPSRC grants EP/S026045/1 and EP/T003553/1, EP/N014588/1, EP/T017961/1, the Wellcome Innovator Awards 215733/Z/19/Z and 221633/Z/20/Z, the European Union Horizon 2020 research and innovation programme under the Marie Skodowska-Curie grant agreement No. 777826 NoMADS, the Cantab Capital Institute for the Mathematics of Information and the Alan Turing Institute.}

\begin{abstract}
We consider the problem of recovering the three-dimensional atomic structure of a flexible macromolecule from a heterogeneous cryo-EM dataset. The dataset contains noisy tomographic projections of the electrostatic potential of the macromolecule, taken from different viewing directions,
and in the heterogeneous case, each cryo-EM image corresponds to a different conformation of the macromolecule.
Under the assumption that the macromolecule can be modelled as a chain, or discrete curve (as it is for instance the case for a protein backbone with a single chain of amino-acids), we introduce a method to estimate the deformation of the atomic model with respect to a given conformation, which is assumed to be known a priori.
Our method consists on estimating the torsion and bond angles of the atomic model in each conformation  as a linear combination of the eigenfunctions of the Laplace operator in the manifold of conformations.
These eigenfunctions can be approximated by means of a well-known technique in manifold learning, based on the construction of a graph Laplacian using the cryo-EM dataset.
Finally, we test our approach with synthetic datasets, for which we recover the atomic model of two-dimensional and three-dimensional flexible structures from simulated cryo-EM images.
\end{abstract}

\date{\today}

\maketitle

\section{Introduction}

One of the central problems in the field of structural biology is that of determining the three-dimensional structure and dynamics of biological macromolecules.
It is well-known that the function of biological macromolecules is determined, not only by the chemical composition, but also by the three-dimensional structure. 
Big macromolecules,  such as proteins, can be composed of thousands of atoms and, despite of the physical and biological knowledge about the formation of the bonds between atoms, determining the three-dimensional configuration of big macromolecules turns out to be an extremely complex task.
Moreover,  most of these biological macromolecules are flexible and may deform their structure, adopting different conformations. 
In this case, providing a single conformation of the macromolecule does not solve the problem of determining the three-dimensional structure,
as one would like to provide a full description of all the possible conformations.

In single particle cryogenic electron microscopy (cryo-EM), an aqueous solution containing the macromolecule of interest is rapidly frozen and then imaged by means of a transmission electron microscope.
Each image (micrograph) contains many samples of the macromolecule (particles) at various (unknown) conformations. During the process known as particle picking,  the particles are individually  selected from the micrograph to produce a series of 2D images, each containing the tomographic projection of the electrostatic potential generated by the single particle and its surrounding buffer.
The fact that the molecular sample may contain different conformations of the macromolecule makes single-particle cryo-EM a well-suited experimental technique to study the structural dynamics of the macromolecule.

The image formation in cryo-EM is typically modelled as a parallel beam ray transform (tomographic projection) of the function representing the electrostatic potential of the particle and surrounding buffer along a particular unknown viewing direction. This is followed by a 2D convolution in the detector plane with a point spread function, which models the microscope optics and detector response.
To formalise the above, let $\{u_1, \ldots , u_n\}$ denote the functions representing the electrostatic potential of each of the particles in the molecular sample.
Note that $u_i$ can be viewed as a volumetric density function depending on the specific conformation and the orientation of each particle.
The corresponding cryo-EM images $\{Y_1, \ldots , Y_n\} \subset \R^{N\times N}$ can be modelled\footnote{We do not include the rotation and the spatial translation of the particle density in the forward operator, as it is usual in the cryo-EM literature. As we will see in the sequel, it is more convenient for us to assume that the orientation and the spatial location of the particle, as well as its specific conformation, are encoded in the particle density function $u_i$. } as
\begin{equation}
\label{cryo-EM forward operator intro}
Y_i = h_i\ast T(u_i) + \varsigma_i
\end{equation}
where $T$ is the digitized parallel beam ray transform taken along the microscope optical axis,
and $h_i\ast$ is the 2D convolution in the image plane with the so-called \textit{point spread function} (PSF) $h_i: \R^2 \to \R$, which is given analytically by its Fourier transform:
\begin{equation}
\label{CTF intro}
\hat{h}_i (\xi) = - A(\xi) \left(  \sqrt{1-\alpha} \sin \left( \dfrac{\Delta z}{2}\lambda |\xi|^2 - \dfrac{C_s}{4}\lambda^3 |\xi|^4 \right) 
+ \alpha \cos \left( \dfrac{\Delta z}{2} \lambda  |\xi|^2 - \dfrac{C_s}{4} \lambda^3 |\xi|^4\right) \right).
\end{equation}
Here, $0<\alpha <1$ is the \textit{amplitude contrast ratio},
 $\Delta z$ is the defocus, $C_s$ is the spherical aberration, and $\lambda$ is the (relativistically corrected) wavelength of the imaging electron.
The function $A:\R^2 \to \R$ represents the aperture function, which is commonly the characteristic function of a disc centred at the optical axis, with radius given by the  objective aperture of the electron microscope.
The function $\hat{h}_i$ in \eqref{CTF intro} is known in the literature as the Contrast Transfer Function (CTF).
Note that, due to the zero crossings of the CTF, the information associated to the frequencies in which $\hat{h}_i$ vanishes is lost.

In \eqref{cryo-EM forward operator intro}, $\varsigma_i$ is typically taken from  $(N\times N)$-dimensional normal distribution, with $N\times N$ denoting the resolution in pixels of the 2D cryo-EM images.
A detailed description of the image formation in cryo-EM can be found in
\cite{frank2006three,vulovic2013image}.
It is worth noting that, in order to avoid the damage of the molecule, the electron dose of the microscope is kept very low, resulting in micrographs dominated by noise (see Figure \ref{fig: 3d struct data}).

\subsection{Related work}

As already indicated, the goal in single particle cryo-EM is to recover the 3D structure of a macromolecule from noisy tomographic projections of single particles, with the difficulty of not knowing the orientation of the particle in each projection.
This problem has attracted a lot of attention in the last decade, and many methods have been proposed to address it, most of them under the assumption of having an homogeneous sample \cite{frank2006three, barnett2017rapid, cheng2015primer, milne2013cryo, vinothkumar2016single}, where all the tomographic projections in the dataset are generated by particles with identical 3D structure, i.e. the macromolecule has a single conformation.
The prevalent method nowadays is the Bayesian approach,  first introduced in \cite{sigworth1998maximum}, and further developed in \cite{scheres2012bayesian}, in which a probability distribution for the viewing directions of the tomographic projections and the density function producing these projections are estimated in an alternating manner.
However,  large macromolecules tend to be flexible, and therefore, the homogeneity assumption does not hold in those cases.
Some methods have been proposed in the last years to recover the structural heterogeneity of macromolecules from single particle cryo-EM data.
Here, we must make a distinction between two types of heterogeneity, namely, discrete heterogeneity, in which the molecular sample contains a small number of different conformations (typically two or three); and continuous heterogeneity, in which the 3D structure varies continuously, and the conformations found in the molecular sample can be seen as an approximation of a continuous low-dimensional manifold.
For the case of discrete heterogeneity, there are several software packages available, such as RELION \cite{scheres2012relion}, cryoSPARC \cite{punjani2017cryosparc}, FREALIGN \cite{lyumkis2013likelihood} and cisTEM \cite{grant2018cistem}, which classify the 2D cryo-EM images in clusters depending on the conformation and estimate the 3D electrostatic potential corresponding to each cluster.

For the case of continuous heterogeneity,  which is the one that we address in this work, an existing approach consists in performing a Principal Component Analysis (PCA) of the 3D electrostatic potentials, represented in an  $N'\times N'\times N'$ voxel grid \cite{penczek2002variance, penczek2011identifying, penczek2006estimation,liao2010classification}, with typically $N'\ll N$.
See also the more recent works \cite{katsevich2015covariance, anden2015covariance, anden2018structural} for a variant of this method,
which has shown to be able to recover the continuous heterogeneity in some molecular samples.
However, as discussed in \cite[subsection 3.3]{moscovich2020cryo}, this method is limited to low-resolution reconstructions of the 3D structure.
More recent works address the continuous heterogeneity by using Deep Neural Networks. For instance,  CryoDRGN \cite{zhong2021cryodrgn} uses a Variational Auto-Encoder (VAE) to estimate the density map corresponding to each conformation in the heterogeneous cryo-EM dataset.
This is further developed in \cite{rosenbaum2021inferring},  which uses a VAE approach to recover the continuous heterogeneity, with the aim of reconstructing, not the volume density, but the atomic model associated to each conformation.

In  \cite{moscovich2020cryo}, an interesting method based on manifold learning is proposed to increase the resolution of 3D reconstructions of continuously heterogeneous macromolecules.
Their idea consists in using the low-resolution density representation of the particles (obtained for instance by the methods in \cite{anden2018structural}) to construct a graph Laplacian,  which is then used in lieu of the Laplace-Beltrami operator  to approximate the spectral properties of the manifold of conformations $\mathcal{M}$.
Then,  a high-resolution volumetric voxel representation of the 3D structure is estimated as a linear combination of the eigenvectors associated to the smallest eigenvalues of the graph Laplacian.
The coefficients in the spectral decomposition of each voxel are estimated by solving a least squares problem, using the forward operator applied to the reconstructed densities, and comparing it with the noisy images in the cryo-EM dataset.
These coefficients give a 3D density map associated to each of the first eigenvalues that, in \cite{moscovich2020cryo}, are referred to as eigenvolumes.

\begin{figure}
\centering
\begin{subfigure}[c]{.35\textwidth}
\centering
\includegraphics[scale=.8]{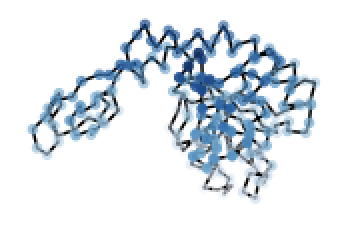}
\end{subfigure}
\hfill
\begin{subfigure}[c]{.6\textwidth}
\includegraphics[scale=.6]{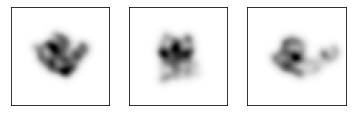}

\includegraphics[scale=.6]{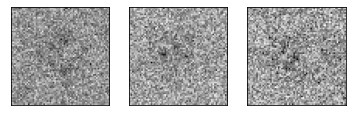}
\end{subfigure}
\caption{At the left we see the coarse-grained model of the adenylate kinase protein representing its backbone in a specific conformation. At the right, we see the 2D projections of the 3D electrostatic potential of three particles from different viewing directions and  different conformations. The images in the first row are the clean projections, and the images in the second row are the same images after a 2D convolution with the PSF and the addition of Gaussian noise.  The signal to noise ratio in these examples is 0.0548.}
\label{fig: 3d struct data}
\end{figure}

\subsection{Our contribution}

In this work, we propose a variant of the method introduced in \cite{moscovich2020cryo}  to address the problem of recovering the 3D structure in a heterogeneous cryo-EM dataset.
The main difference with respect to \cite{moscovich2020cryo} is that, instead of estimating a volumetric representation of the macromolecule,  we rather aim at recovering the atomic model corresponding to each conformation, i.e. the 3D configuration of the atoms.
In order to apply our method we need to make two important assumptions concerning the prior knowledge about the structure of the macromolecule:
\begin{enumerate} 
\item We consider that the backbone of the macromolecule is composed by a fixed number of atoms forming a chain, or discrete curve,  in which the distance between adjacent atoms is constant.
\item We assume that we have access to the atomic model of the macromolecule corresponding to a specific conformation.
\end{enumerate}

Regarding the first assumption,  although most of the macromolecules have more complex structures,  many of them, as for instance proteins, can be approximated by means of the so-called backbone, a discrete curve determining the positions of the C-$\alpha$ atoms (see Remark \ref{rmk: on the chain molecular structure}). 
As for the second assumption, the atomic model for many macromolecules is available in existent databases.
For the particular case of proteins, the recently developed deep learning based approach AlphaFold \cite{jumper2021highly, tunyasuvunakool2021highly, varadi2022alphafold} can provide highly accurate predictions of the 3D structure of proteins from their primary sequence (sequence of amino-acids).
However, the structures in the databases typically contain the information regarding the macromolecule in a single conformation.
In this work,  our goal is, indeed, to use the cryo-EM data to estimate  the deformation of the atomic model in each particle of the molecular sample with respect to the known conformation (see Figure \ref{fig: atomic model pred intro}).

\begin{remark}
\label{rmk: on the chain molecular structure}
An atomic model in which the atoms form a discrete curve can be used in practice as an approximation of the 3D structure of a protein with a single chain of amino-acids.
A protein is a macromolecule structure consisting of one or more chains of amino-acid residues, which are connected by peptide bonds, forming the so-called backbone.
The central atom in each amino-acid residue is known as the C-$\alpha$ atom, and connects the side chain of the amino-acid residue to the backbone.
One way to approximate the 3D structure of the protein is by determining the C-$\alpha$ positions conforming the backbone, which form, indeed, a discrete curve satisfying our assumption \cite{hu2011discrete}.
\end{remark}

In the following, we offer a more formal description of our approach.
See Figure \ref{fig: diagram} for a diagram outlining our method.
The main idea  is to use prior knowledge about the atomic structure to introduce a parametrisation of the space of atomic models satisfying such prior.
In our case,  we assume that the molecule of interest can be approximated by its backbone, which can be modelled as a discrete curve,  i.e.  a sequence of points $\mathbf{z}:= \{ z_1, z_2, \ldots , z_m\} \in \R^{3m}$, where $m\in \N$ denotes the number of atoms in the model.
As we will see in subsection \ref{subsec: structure spectral decomp}, such 3D structures can be represented, up to translations and rotations, by the \emph{torsion} and \emph{bond} angles  at each point of the discrete curve, that we denote by 
$\Theta = (\theta_1, \theta_2,\ldots, \theta_{m-2})\in \R^{m-2}$ and $\Psi = (\psi_1, \psi_2,\ldots , \theta_{m-2})\in \R^{m-2}$ respectively.
In order to determine the conformation of the macromolecule for each particle $u_i$, we estimate the parameters $\Theta_i$ and $\Psi_i$ as 
\begin{equation}
\label{spectral decomp intro}
\Theta_i = \Theta_i (A) : = \Theta_0 +  \sum_{k=0}^{K-1} \mathbf{a}_k \phi_i^{(k)}
\qquad \text{and} \qquad
\Psi_i = \Psi_i(B) := \Psi_0 +  \sum_{k=0}^{K-1} \mathbf{b}_k \phi_i^{(k)},
\end{equation}
where the vectors $\Theta_0\in \R^{m-2}$ and $\Psi_0\in \R^{m-2}$ represent the torsion and bond angles of the atomic model in the known conformation, and the coefficient matrices 
$$
A = [\mathbf{a}_0, \mathbf{a}_1, \ldots , \mathbf{a}_{K-1}] \in \R^{(m-2)\times K} \qquad \text{and} \qquad B=[\mathbf{b}_0, \mathbf{b}_1, \ldots , \mathbf{b}_{K-1}] \in \R^{(m-2)\times K} 
$$ 
are to be estimated from the noisy cryo-EM images. 
Note that each column in the matrices $A$ and  $B$ is a vector of coefficients $\mathbf{a}_k\in \R^{m-2}$ and $\mathbf{b}_k \in \R^{m-2}$ respectively.
Following the same ideas as in \cite{moscovich2020cryo},  the vectors $\{ \phi^{(k)}\}_{k=0}^{K-1}$ defined as
$$\phi^{(k)} = \left(\phi_1^{(k)}, \phi_2^{(k)}, \ldots , \phi_n^{(k)}\right)\in \R^n \qquad  \forall k\in \{0,1,\ldots, K-1\}$$ 
that we use in \eqref{spectral decomp intro} are the eigenvectors associated to the $K$ smallest eigenvalues of a graph Laplacian matrix, which is used to approximate the spectral properties of the unknown manifold of conformations $\mathcal{M}$ (see Figures \ref{fig: eigenvec 2d}, \ref{fig: eigenvec 3d} and \ref{fig: eigenvec 3d 2} for an illustration). 
The underlying idea, as presented in \cite{moscovich2020cryo},  is that each node in the graph represents an image of the dataset, and the weights of the edges between the nodes represent the similarity between the conformations of the underlying particles.
The eigenvectors of the graph Laplacian are then used as an approximation of the eigenfunctions of the Laplace-Beltrami operator defined on the (unknown) manifold of conformations.
See subsection
\ref{subsec: structure spectral decomp} for further details about the spectral decomposition of the  atomic model.

A critical point in this approach is the construction of the aforementioned weighted graph, or more precisely, the way in which one estimates the similarity between the conformation in each cryo-EM image.
One way to do so is by using a low-dimension representation of the particles, which can be obtained from different existing techniques. Namely, by means of a low-resolution PCA reconstruction \cite{anden2018structural}, or the so-called latent space representation obtained from a variational autoencoder \cite{rosenbaum2021inferring,zhong2021cryodrgn}.
See Appendix  \ref{appdx: manifold spectral representation} for a discussion about the construction of the graph of similarities.

The number $K\in \mathbb{N}$ of eigenvectors in the spectral decomposition is a hyper-parameter that can be chosen depending on the flexibility of the structure, or the type of deformations that we aim to recover.
Since we use only the terms in the spectral decomposition associated to the smallest eigenvalues of the graph Laplacian, our method is able to capture only the low-frequency deformations of the backbone. The deformations of high frequency, typically associated to small vibrations of the atoms, are not recovered by our method (see Remark \ref{rmk: high ferquencies}).
In our numerical experiments in section \ref{sec: numerical experiments} we have used $K=10$ and $K=20$ respectively.

\begin{figure}
\centering
\begin{subfigure}[c]{.3\textwidth}
\centering
\includegraphics[scale=.6]{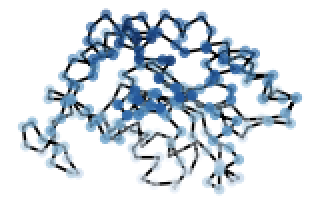}
\caption{Known conformation}
\end{subfigure}
\hfill
\begin{subfigure}[c]{.3\textwidth}
\centering
\includegraphics[scale=.6]{Figures_new/fig_3d_1/ground_truth.png}
\caption{Ground truth}
\end{subfigure}
\hfill
\begin{subfigure}[c]{.3\textwidth}
\centering
\includegraphics[scale=.6]{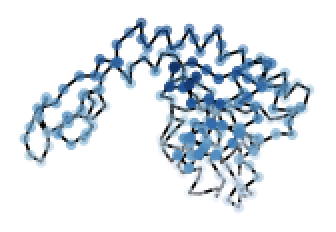}
\caption{Estimated atomic model}
\end{subfigure}
\caption{At the left, the backbone of the macromolecule in the known conformation, which is used as initial guess for all the particles. In the middle, the specific conformation of one of the particles of the synthetic dataset. At the right, the atomic model estimated by means of our method.}
\label{fig: atomic model pred intro}
\end{figure}

Finally, we need to estimate the coefficients $\mathbf{a}_k$ and $\mathbf{b}_k$ in the spectral decomposition \eqref{spectral decomp intro}.
This is done in subsection \ref{subsec: tomographic reconstruction} by means of a minimisation problem,  evaluating the data fidelity of the estimated atomic model for each particle with the corresponding cryo-EM image.
Let us briefly describe how this optimisation problem is formulated.
Under the assumption that the 3D structure can be described as a discrete curve in $\R^3$,  and that the spatial location and orientation of the particles, represented by $\{(\hat{z}_{i},\hat{F}_{i})\}_{i=1}^n \in [\R^3 \times SO(3)]^n$, are known\footnote{In the original cryo-EM problem, the orientation of the particles is not available in general, however, since we are interested here on recovering the structural heterogeneity, we may assume that the orientations have been accurately estimated, previously, by means of existent methods as for instance RELION \cite{scheres2012relion}}, we estimate the 3D density associated to the $i$-th particle as a sum of Gaussian densities centred at the positions of the atoms in the $i$-th particle, i.e.
$$
x\in \R^3 \longmapsto \hat{U}_i (A,B) (x) = \sum_{z \in \Gamma_i (A,B)} \nu \exp \left[ - \dfrac{\| x- z\|^2}{2\sigma^2} \right]. 
$$ 
Here, $\nu, \sigma>0$ are hyperparameters to be chosen a priori,  and, 
for each $i\in \{ 1,2, \ldots , n\}$ the point cloud $ \Gamma_i (A,B) = \mathbf{Z} (\Theta_i (A), \Psi_i (B), \hat{z}_{i}, \hat{F}_{i})\in \R^{3m}$ are the estimated atom positions of the $i$-th particle, obtained as the solution of\footnote{The product $e_3 F_j$ denotes the matrix multiplication of the row vector $e_3 = (0,0,1)$ by the square matrix $F_j\in SO(3)$. Hence, $e_3F_j$ simply represents the third row of the matrix $F_j$.}
$$
\begin{cases}
z_{j+1} = z_j + \delta e_3 F_j & j\in \{ 1, \ldots , m-1\} \\
F_j = R (\theta_{ij}(A), \psi_{ij}(B)) F_{j-1} & j\in \{ 2, \ldots , m-2\} \\
\text{with} \ z_{j_0}=\hat{z} \in \R^3 \ 
\text{and} \ F_{j_0} = \hat{F} \in SO(3),
\end{cases}
$$
where $R(\theta, \psi)\in SO(3)$ is the rotation matrix of angles $\theta$ and $\psi$ (see Appendix \ref{appdx: atomic model param} for further details).
The functions 
$$
\Theta_i(\cdot) = [\theta_{i1} (\cdot), \theta_{i2}(\cdot), \ldots]\in \R^{m-2}
 \quad \text{and} \quad 
 \Psi_i (\cdot) =  [\psi_{i1} (\cdot), \psi_{i2}(\cdot), \ldots]\in \R^{m-2}
 $$ in the definition of $\Gamma_i (A,B)$ are given by \eqref{spectral decomp intro}.
 Then, the estimated cryo-EM images are obtained by applying the forward operator to the estimated volume densities, i.e.
 $$
 \hat{Y}_i (A,B) =  h_i \ast T [\hat{U}_i (A,B)] \in \R^{N\times N}, \qquad \text{for} \ i\in \{1, \ldots , n\},
 $$
 where $T(\cdot)$ is the parallel beam ray transform and $h_i$ is the PSF.
We can now formulate the following least squares problem to estimate the parameters 
$$
A = [\mathbf{a}_0, \mathbf{a}_1,  \ldots , \mathbf{a}_{K-1}]\in \R^{(m-2)\times K}
\quad \text{and} \quad
B = [\mathbf{b}_0, \mathbf{b}_1, \ldots , \mathbf{b}_{K-1}]\in \R^{(m-2)\times K}
$$ 
in \eqref{spectral decomp intro} as
$$
[\hat{A} , \hat{B}] :=
\underset{A, B}{\operatorname{argmin}}  \dfrac{1}{n} \sum_{i=1}^n \| \ \hat{Y}_i (A,B) - Y_i \|^2,
$$
where $\{ Y_i\}_{i=1}^n$ are the 2D images in the cryo-EM dataset.

This is a non-convex minimisation problem, and we may use stochastic gradient descent (SGD) to approximate a solution.
The SGD algorithm is implemented by randomly dividing the set of images in batches and, at each step, computing the gradient using only the data in one of the batches.  This procedure is iterated through several epochs, after which the partition in batches is re-shuffled.
As it is usual when applying iterative methods to approximate the solution of a non-convex problem, the success of the method heavily depends on the initialisation of the parameters, in this case the vectors $\mathbf{a}_k$ and $\mathbf{b}_k$ in \eqref{spectral decomp intro}.
Here is where the known conformation of the macromolecule, represented by the parameters $(\Theta_0, \Psi_0)$, becomes very important.
Indeed, for the initialisation of the iterative method, we  simply set all the parameters $\mathbf{a}_k$ and $\mathbf{b}_k$ to be equal to $0$, so that the prediction of the atomic model for all the particles at the initialisation is the known conformation.
This provides a reasonable initial guess for the atomic model of each particle, which is then deformed, changing the conformation of the macromolecule, in order to improve the predictions to better fit the cryo-EM dataset.

In section \ref{sec: numerical experiments} we present several numerical experiments in which our method is used to reconstruct the heterogeneous atomic model of flexible 2D and 3D structures from noisy tomographic projections.  
Although it is well-known that the tomographic reconstruction with unknown viewing directions presents fundamental differences (concerning the uniqueness of solution) between the 2D \cite{basu2000uniqueness,basu2000feasibility} and the 3D \cite{kurlberg2021formal,van1987angular} setting, these differences are not relevant in our numerical experiments, in which we assume that the pose of the particle in each cryo-EM image has been estimated a priori, and moreover, our reconstruction is based on estimating atom positions, rather than estimating the underlying density. 
In the first experiment (subsection \ref{subsec: example 2d}), we consider a 2D structure in which the discrete curve forms a flexible box with two moving arms. In this case, the density functions associated to the atomic structure are 2D images, and the associated tomographic projections are therefore 1D functions (see figure \ref{fig: 2d struct data}).
In subsection \ref{subsec: example3d},  we apply our method to recover the continuous heterogeneity of two 3D structures from an initial conformation.  In both cases,  a synthetic dataset is generated, in which the ground truth  of the different conformations are simulated trajectories using Molecular Dynamics (MD).
In the first experiment (subsection \ref{subsubsec: example 3d 1}),
the MD simulation is taken from \cite{Beckstein2018}, and corresponds to a closed-to-open transition of the protein adenylate kinase.
In the second experiment (subsection \ref{subsubsec: example 3d 2}) the MD simulation is taken from \cite{shaw2020molecular}, an represents a $25\mu s$ trajectory of SARS-CoV-2 helicase nsp13 starting from open conformation.

The code used for the numerical experiments in this paper is publicly available at 

\href{https://github.com/carlosesteveyague/Atomic-structure-reconstruction-for-cryoEM}{https://github.com/carlosesteveyague/Atomic-structure-reconstruction-for-cryoEM}

The rest of the paper is structured as follows:
\begin{enumerate}
\item In section \ref{sec: the method} we describe our method in detail.
First, we introduce in subsection \ref{subsec: structure spectral decomp} the spectral decomposition of atomic structures satisfying the aforementioned discrete curve property.
In subsection \ref{subsec: approx spectr decomp}, we describe the method, based on manifold learning, to approximate the spectral decomposition by means of the cryo-EM dataset.
In subsection \ref{subsec: tomographic reconstruction}, we formulate the minimisation problem which is used to estimate the coefficients in the spectral decomposition.
\item In section \ref{sec: numerical experiments} we present some numerical experiments in which we use our method to carry out the tomographic reconstruction of 2D and 3D atomic structures.
\item In Section \ref{sec: conclusions and future}, we sum-up the conclusions of our work and describe the possible steps to be taken in view of using our method as one of the steps in the 3D reconstruction pipeline of heterogeneous macromolecules from real cryo-EM data. 
\item At the end, we include Appendix \ref{appdx: atomic model param}, where we describe the parametrisation of discrete curves, based on a discrete version of the Frenet frames \cite{hu2011discrete}, that we use throughout the paper; and Appendix \ref{appdx: manifold spectral representation}, where we describe the methods, taken from \cite{anden2018structural} and \cite{moscovich2020cryo}, to construct a low-dimension representation of the particles from the cryo-EM dataset, which is then used to build the graph Laplacian that we use in our method.
\end{enumerate}

\section{The method}
\label{sec: the method}

In this section we describe our method in detail. See figure \ref{fig: diagram} for a high-level description of the method.
We start by introducing the spectral decomposition of the atomic structure (subsection \ref{subsec: structure spectral decomp}).
Then, we describe the method to approximate this spectral decomposition (subsection \ref{subsec: approx spectr decomp}).
Finally, we formulate the minimisation problem that we use to estimate the coefficients in the spectral decomposition (subsection \ref{subsec: tomographic reconstruction}).

\begin{figure}
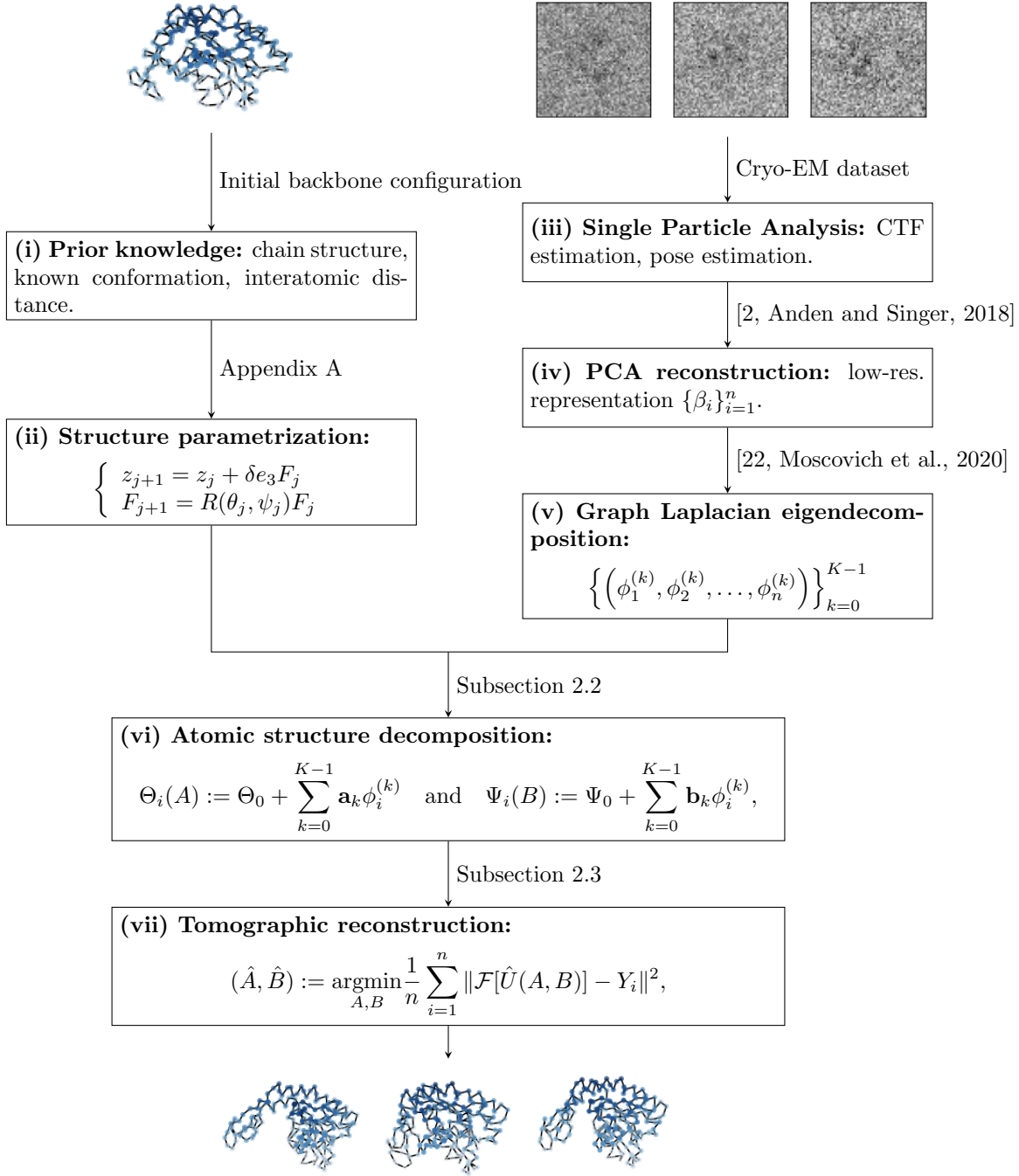

\centering
\include{diagram/cryo-EM_diagram}
\caption{Diagram of our method. We observe that it combines prior biological knowledge about the macromolecule (in the left column) and the experimental data from cryo-EM (in the right column). The atomic structure decomposition in (vi) is constructed using the prior knowledge, whereas the coefficients are estimated by solving the minimisation problem in (vii), using the cryo-EM images $\{ Y_i \}_{i=1}^n$.}
\label{fig: diagram}
\end{figure}

\subsection{Spectral decomposition of the atomic structure}
\label{subsec: structure spectral decomp}

As mentioned in the introduction, our approach relies on the fact that we have prior knowledge about the structure of the molecule. 
This knowledge consists on knowing the number of atoms of the backbone and having access to the 3D structure of the macromolecule in some specific conformation (this one can be obtained by using for instance AlphaFold \cite{jumper2021highly}).
These assumptions will be crucial later on to parametrise the set of possible atomic models and carry out the spectral decomposition in the parameter space.
In particular,  we make the following structural assumption:

\begin{assumption}
\label{assumption: discrete curve}
The backbone of the macromolecule has a fixed number of $m\in \N$ atoms forming a chain (or discrete curve), i.e.  there exists $\delta>0$ such that the point cloud $\{ z_1, z_2, \ldots , z_m\}\in \R^{3m}$ representing the positions of the atoms satisfies
$$
\| z_{j+1} - z_j\| \approx \delta, \qquad \forall j\in \{1, \ldots , m-1\}. 
$$
Moreover, we assume that the set of conformations of the macromolecule, that will be denoted by $\mathcal{M}$, is a connected compact manifold of dimension  $d\geq 1$, that can be embedded in a higher dimensional Euclidean space $\R^{q}$, with $q\geq d$.
\end{assumption}

Under this assumption, and considering that we know the number of atoms $m\in \N$ and the distance $\delta >0$ between adjacent atoms,  we will see in the appendix \ref{appdx: atomic model param} that the atomic model associated to any conformation can be represented as the solution of a dynamical system of the form
\begin{equation}
\label{dynamical system}
\begin{cases}
z_{j+1} = z_j + \delta e_3 F_j & j\in \{ 1, \ldots , m-1\} \\
F_{j+1} = R (\theta_j, \psi_j) F_j & j\in \{ 1, \ldots , m-2\} \\
\text{with} \ z_{j_0} = \hat{z} \in \R^3 \ 
\text{and} \ F_{j_0} = \hat{F} \in SO(3).
\end{cases}
\end{equation}
where $e_3 = (0,0,1)\in \R^3$, the sequence $\{z_1, \ldots , z_m  \}\subset \R^{3m}$ denotes the positions of the atoms in the macromolecule, and $\{ F_1, \ldots , F_{m-1}\} \subset SO(3)^{m-1}$ are the so-called Discrete Frenet Frames (see the appendix \ref{appdx: atomic model param} and also \cite{hu2011discrete} for further details). The inter-atomic distance $\delta>0$ is fixed and given,  $R (\theta_j, \psi_j)\in SO(3)$ denotes the rotation matrix in $\R^3$ with angles $\theta_j$ and $\psi_j$, i.e.
$$
R (\theta_j, \psi_j) = 
\left(\begin{matrix}
\cos \psi_j \cos \theta_j & \cos \psi_j \sin \theta_j & - \sin \psi_j \\
-\sin \theta_j & \cos \theta_j & 0 \\
\sin \psi_j \cos \theta_j & \sin \psi_j \sin \theta_j & \cos \psi_j
\end{matrix}\right),
$$
and the index $j_0 \in \{1, \ldots, m-1\}$ in \eqref{dynamical system} is chosen a priori (typically in the middle of the chain).
We observe that the space of such structures can be parametrised by the rotation angles $(\theta_j,\psi_j)$ at each point of the discrete curve, along with the initial condition $(\hat{z}, \hat{F})$, which determines the spatial position and the orientation of the particle.
Let us define the parameter space of such discrete curves as
\begin{equation}
\label{param atomic model}
(\Theta,  \Psi,  \hat{z},  \hat{F}) \in \mathcal{A} \times \mathcal{I} := \left( [-\pi, \pi]^{m-2}\times [-\pi,\pi]^{m-2} \right)  \times \left( \R^3 \times SO(3)\right),
\end{equation}
where $\Theta = (\theta_1, \ldots, \theta_{m-2})$ and $\Psi = (\psi_1, \ldots , \psi_{m-2})$
represent the torsion and bond angles at each point in the curve, 
and $\hat{z}$ and $\hat{F}$ determine the position and the orientation of the structure at the reference atom $j_0$, that is chosen a priori, typically in the central part of the macromolecule.
See more details about the parametrisation of the atomic model in the appendix \ref{appdx: atomic model param}.

The choice of the reference atom is arbitrary, but it is convenient to choose a point in a part of the structure which is invariant in all the conformations, i.e. we want $z_{j_0}$ and $F_{j_0}$ to depend only on the location and the orientation of the particle but not on the conformation.
Choosing as reference an atom in a flexible part of the macromolecule would imply that the parameters $\hat{z}$ and $\hat{F}$ are also a function of the conformation, and therefore, these parameters should also be included in the spectral decomposition \eqref{spectral decomp intro}.

In the parametrisation  \eqref{dynamical system}-\eqref{param atomic model} of the atomic model we observe, on one hand, that the parameters $\hat{z}$ and $\hat{F}$ only determine the spatial location and the orientation of the particle.
On the other hand,  $\Theta$ and $\Psi$, which are invariant under rotations and translations of the particle, determine the 3D structure up to rotations and translations.
Roughly speaking, they determine the shape of the curve.
The parameters $\Theta$ and $\Psi$  are therefore the relevant parameters to describe the structural heterogeneity of the macromolecule, as modifying these parameters translates into deformations of the shape of the curve. Moreover, we stress that these deformations are translation and rotation invariant.
In our numerical experiments in section \ref{sec: numerical experiments}, we shall assume the knowledge of the location and the orientation of the particles in the cryo-EM dataset, represented by the knowledge of $\hat{z}$ and $\hat{F}$ for each of the particles, so that the tomographic reconstruction reduces to estimate the parameters $\Theta$ and $\Psi$ for each particle.

Denoting by $\mathcal{M}$ the set of possible conformations, our goal is to construct a map
\begin{equation}
\label{conformation to atomic model map}
\eta \in \mathcal{M} \longmapsto [\Theta(\eta) , \Psi(\eta) ] \in 
\mathcal{A} := [-\pi, \pi]^{m-2}\times [-\pi,\pi]^{m-2},
\end{equation}
which maps any element in $\mathcal{M}$ to the associated parameters $\Theta$ and $\Psi$, which determine a unique atomic model satisfying Assumption \ref{assumption: discrete curve}, up to rotations and translations.
As mentioned in the introduction,  a key reason of the success of our method relies on having access to the atomic model of the macromolecule in some specific conformation. This conformation will be used to initialise the iterative algorithm for the tomographic reconstruction of the 3D structure of each particle.

\begin{assumption}
\label{assumption: known conformation}
We consider that we know the atomic model of a single conformation of the macromolecule, and therefore, we can extract the number of atoms $m\in \N$ in the chain, the distance between adjacent atoms $\delta>0$, and compute the torsion and bond angles denoted by $(\Theta_0, \Psi_0)\in \mathcal{A}$.
\end{assumption}

As indicated in the introduction, we are interested in estimating the difference between the known atomic model and the atomic model associated to each conformation, i.e.,  $\Theta(\eta) - \Theta_0$ and $\Psi(\eta) - \Psi_0$ for every $\eta \in \mathcal{M}$.
Following a similar idea as in \cite{moscovich2020cryo}, 
we construct the functions $\Theta(\cdot)$ and $\Psi(\cdot)$ as the linear combination of a finite number of elements in a basis of $L^2(\mathcal{M})$ as follows:
\begin{equation}
\label{spectral decomp}
\Theta (\eta) = \Theta_0 +  \sum_{k=0}^{K-1} \mathbf{a}_k \phi_k(\eta)
\qquad \text{and} \qquad
\Psi (\eta) = \Psi_0 +  \sum_{k=0}^{K-1} \mathbf{b}_k \phi_k (\eta),
\end{equation}
where $K\in \mathbb{N}$ is chosen a priori, $\phi_k(\cdot):\mathcal{M} \to \R$, for $k\in \{0,1, \ldots, K-1\}$ are the eigenfunctions associated to the $K$ smallest eigenvalues of the Laplace-Beltrami operator in $\mathcal{M}$,
and $\mathbf{a}_k\in \R^{m-2}$ and $\mathbf{b}_k\in \R^{m-2}$ for $k\in \{0, 1, \ldots , K-1 \}$ are the coefficients of the spectral decomposition of the atomic structure.

\begin{remark}
\label{rmk: initial conformation}
We see in \eqref{spectral decomp} that we are actually estimating the deformation of the 3D structure with respect to the given known conformation.
As we will see in the sequel, this assumption about the known conformation is extremely important in the estimation of the coefficients $\mathbf{a}_k$ and $\mathbf{b}_k$ in \eqref{spectral decomp}.
Indeed, the estimation of $\mathbf{a}_k$ and $\mathbf{b}_k$ is addressed by means of SGD applied to a non-convex minimisation problem.
As initialisation we take all the coefficients $\mathbf{a}_k$ and $\mathbf{b}_k$ to be equal to $0$, so that the initial guess of the atomic model for any particle coincides with the known conformation.
This produces reasonable predictions for all the particles that, by applying SGD, are then deformed to fit the cryo-EM data and  obtain the atomic model approximating the conformation of each particle.
\end{remark}

\begin{remark}
\label{rmk: high ferquencies}
 We recall, from Assumption \ref{assumption: discrete curve},  that $\mathcal{M}$ is a connected compact manifold, and hence, the eigenfunctions of the Laplace-Beltrami operator in $\mathcal{M}$ form a complete orthonormal basis of $L^2(\mathcal{M})$.
Moreover, it is well-known \cite{berger1971spectre, grebenkov2013geometrical} that the 
eigenvalues $0= \lambda_0 < \lambda_1 \leq \lambda_2 \leq \ldots$ give an estimate of the regularity of the associated eigenfunctions $\{\phi_0(\cdot), \phi_1(\cdot),  \phi_2(\cdot), \ldots\}$.
This can be interpreted as follows:
in the expansion of any smooth function in $\mathcal{M}$ as a linear combination of the eigenfunctions, the terms associated to higher eigenvalues correspond to higher frequencies of the function.
By using only the eigenfunctions associated to the smaller eigenvalues in the spectral decomposition \eqref{spectral decomp},  we only capture deformations of the backbone of low frequency in $\mathcal{M}$.
We note that in this work, we aim at describing only the large deformations of the macromolecule, and we are not interested in high-frequency deformations, which are typically associated to small vibrations of the atoms.
\end{remark}

\begin{remark} 
Another interesting feature of our approach is the fact that we can use prior knowledge about the macromolecule to reduce the number of coefficients to be estimated in \eqref{spectral decomp intro}.
Some parts of the macromolecules may be known to be very stable, so that their internal structure is preserved in all the conformations. This is the case, for instance, of the so-called secondary structures appearing in proteins ($\alpha$-helices and $\beta$-strands).
These substructures only suffer rigid transformations (rotations and translations) from one conformation to another, and then the functions $\theta_i(\eta)$ and $\psi_i(\eta)$ associated to the atoms conforming these substructures should be constant in  $\mathcal{M}$.
This can be achieved by using only the first term in the expansion \eqref{spectral decomp intro} (i.e. we set the parameters $a_{k,i} =0$ and $b_{k,i} = 0$ for all $k \geq  1$), since we known that the first eigenfunction $\phi_0(\cdot)$ is constant in $\mathcal{M}$.
An intermediate situation can be also considered, in which one may use a different number of eigenfunctions in the expansion \eqref{spectral decomp intro} (by setting the corresponding coefficients to zero), depending on the prior knowledge about the flexibility of certain parts of the macromolecule.
See the numerical experiments in subsections \ref{subsubsec_ with prior} and \ref{subsubsec_ no prior} for an illustration of the feature described in this remark.
\end{remark}

\subsection{Approximated spectral decomposition}
\label{subsec: approx spectr decomp}

In single particle cryo-EM, we have access to a dataset with noisy tomographic projections of a macromolecule in different conformations, and taken from different viewing directions.
This is the data that we want to use in order to determine the atomic model of the macromolecule in all its conformations by approximating the spectral decomposition introduced in subsection \ref{subsec: structure spectral decomp}.
As outlined above,  under Assumption \ref{assumption: discrete curve}, the problem consists on
estimating the map $[\Theta (\cdot), \Psi(\cdot)]$ defined in \eqref{conformation to atomic model map}.
In our approach, the functions $\Theta(\cdot)$ and $\Psi(\cdot)$ are estimated by means of an expansion series of the form \eqref{spectral decomp}.
We see that one needs two ingredients: the eigenfunctions of the Laplace-Beltrami operator in $\mathcal{M}$, and the vector coefficients $\mathbf{a}_k\in \R^{m-2}$ and $\mathbf{b}_k\in \R^{m-2}$, where $m$ is the number of atoms in the chain.
The estimation of the latter will be addressed in the following subsection.
For the former, we have the difficulty of not knowing the manifold $\mathcal{M}$,
however, we can apply the same strategy as in \cite{moscovich2020cryo},  based on manifold learning, to obtain a spectral approximation of $\mathcal{M}$ by means of a suitable graph Laplacian, constructed from the cryo-EM dataset.
This method builds on the assumption that we have a low-dimension representation of the conformation of each particle in the dataset.

\begin{assumption}
\label{assumption: low-dimension reconstruction}
For each image in the cryo-EM dataset $\{ Y_i \}_{i=1}^n$ we have a low-dimension representation $\{ \beta_i\}_{i=1}^n \subset \R^q$ of the conformation of the particle in each image, so that the manifold of conformations  $\mathcal{M}$ can be embedded in $\R^q$.
Moreover, this representation is independent of the viewing direction of the tomographic projection, so that $\{ \beta_i\}_{i=1}^n $ can be seen as a discrete approximation of the manifold $\mathcal{M}$.
\end{assumption}

This low-dimension representation of the conformation in each cryo-EM image can be obtained by using existent methods in heterogeneous reconstruction. One possibility is to follow \cite{moscovich2020cryo}, in which the low-dimension representation of the heterogeneity consists on the PCA coordinates of the low-resolution reconstruction of the particles obtained by the method in \cite{anden2018structural}.
See  Appendix \ref{appdx: manifold spectral representation} for a description of the method to obtain the PCA coordinates associated to the conformation of the particle in each image.
Another possibility is to use the latent space variables learnt in a Variational Autoencoder \cite{rosenbaum2021inferring, zhong2021cryodrgn}, in which the representation of each image in the latent variables encodes the information about the corresponding conformation.
For the independence of $\beta_i$ with respect to the viewing direction,  this assumption can be justified by assuming that the reconstruction from \cite{anden2018structural} is accurate independently of the orientation of the particle.
Now, using this low-dimension representation of the dataset, we construct a weighted graph with the similarities between the conformations in every two particles, and compute the eigenvectors associated to the smallest eigenvalues of the symmetric normalised Laplacian matrix associated to the graph of similarities.

Let $K\in \N$ be fixed,  let $L\in \mathcal{S}(n)$ be the symmetric graph Laplacian matrix associated to the aforementioned matrix of similarities, and let 
\begin{equation}
\label{matrix eigenvectors}
\Phi = [\phi^{(0)} | \phi^{(1)} | \cdots | \phi^{(K-1)}]\in \R^{n\times K}
\end{equation}
be the matrix which has, as columns, the eigenvectors $\phi^{(k)}\in \R^n$ associated to the $K$ smallest   eigenvalues of the graph Laplacian matrix $L$.
Some definitions of the graph Laplacian are known to converge to a continuous operator in $\mathcal{M}$, in the sense that the eigenvectors $\phi^{(k)}$ of the graph Laplacian converge to the eigenfunctions $\phi_k (\cdot)$ of this operator \cite{lee2016spectral, rosasco2010learning, von2008consistency}.
In some cases, the limit operator is the Laplace-Beltrami operator in $\mathcal{M}$, but actually, we only need that the eigenfunctions of the limit operator form an orthonormal basis of $L^2(\mathcal{M})$.

As in \cite[Section 5.3]{moscovich2020cryo},  we assume that the ordered eigenvectors  $\phi^{(0)} , \phi^{(1)}, \ldots \in \R^n$ of the graph Laplacian converge in probability to a set of eigenfunctions $\phi_0(\cdot), \phi_1 (\cdot), \ldots : \mathcal{M} \to \R$ of some continuous linear differential operator on $\mathcal{M}$,  in the sense that, for all $k\in \mathbb{N}$, we have
$$
\sup_{i=1, \ldots , n} | \sqrt{n} \phi_i^{(k)} - \phi_k (\beta_i) | \to 0 \qquad
\text{as} \quad n\to \infty. 
$$
Furthermore, we assume that this set of eigenfunctions $\{ \phi_k(\cdot)\}_{k=0}^\infty$ form an orthonormal basis of $L^2(\mathcal{M})$.
Hence, for each particle $u_i$ in the cryo-EM dataset, 
we can estimate the parameters $[\Theta_i, \Psi_i]$ of the corresponding conformation by approximating the expression \eqref{spectral decomp} as
\begin{equation}
\label{spectral decomp approx}
\Theta_i (A) =
\Theta_0 + \sum_{k=0}^{K-1} \mathbf{a}_k \phi_i^{(k)}
= \Theta_0 + A \Phi_i
\qquad \text{and} \qquad
\Psi_i (B) =
\Psi_0 + \sum_{k=0}^{K-1} \mathbf{b}_k \phi_i^{(k)}
= \Psi_0 + B \Phi_i
\end{equation}
where $A = [\mathbf{a}_0, \mathbf{a}_1, \ldots, \mathbf{a}_{K-1}] \in \R^{(m-2)\times K}$ and $B = [\mathbf{b}_0, \mathbf{b}_1, \ldots, \mathbf{b}_{K-1}] \in \R^{(m-2)\times K}$ are the matrices of coefficients, and
$\Phi_i\in \R^K$, for $i=1, \ldots,n$, are the column vectors with the $i$-th component of each eigenvector from \eqref{matrix eigenvectors}, i.e.
$$
\Phi_i = (\phi_i^{(0)}, \phi_i^{(1)} , \ldots , \phi_i^{(K-1)}) \in \R^K.
$$

\begin{figure}
\centering
\begin{subfigure}{.24\textwidth}
\centering
\includegraphics[scale=.4]{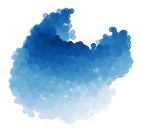}
\caption{$[\phi^{(1)}, \phi^{(2)}, \phi^{(3)}]$}
\end{subfigure}
\hfill
\begin{subfigure}{.24\textwidth}
\centering
\includegraphics[scale=.4]{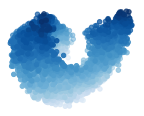}
\caption{$[\phi^{(1)}, \phi^{(2)}, \phi^{(4)}]$}
\end{subfigure}
\hfill
\begin{subfigure}{.24\textwidth}
\centering
\includegraphics[scale=.4]{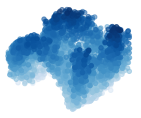}
\caption{$[\phi^{(1)}, \phi^{(3)}, \phi^{(4)}]$}
\end{subfigure}
\hfill
\begin{subfigure}{.24\textwidth}
\centering
\includegraphics[scale=.4]{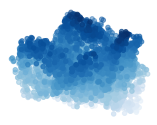}
\caption{$[\phi^{(2)}, \phi^{(3)}, \phi^{(4)}]$}
\end{subfigure}
\caption{Representation of the eigenvectors $\phi^{k}\in \R^n$ of the graph Laplacian constructed from the synthetic dataset of 2D structures presented in subsection \ref{subsec: example 2d}.  Each plot can be seen as an approximation of the manifold of conformations $\mathcal{M}$ projected in different three-dimensional linear spaces.}
\label{fig: eigenvec 2d}
\end{figure}

\subsection{Tomographic reconstruction}
\label{subsec: tomographic reconstruction}

In this subsection, we use the elements presented in the preceding subsections to carry out the reconstruction of the 3D structure corresponding to the conformation in each of the cryo-EM images. This is done by means of a minimization problem, in which the coefficients $A$ and $B$ in \eqref{spectral decomp} are estimated using the data fidelity with the cryo-EM images.

We recall that, under the structural Assumption \ref{assumption: discrete curve}, the backbone forms a discrete curve,  which can be represented as the solution to \eqref{dynamical system},
and therefore,  for each image $\{Y_i\}_{i=1}^n$, we need to estimate the parameters
$$
(\Theta_i , \Psi_i) \in \mathcal{A}:=  [-\pi , \, \pi]^{m-2}\times [-\pi , \, \pi]^{m-2}
$$
representing the torsion and bond angles of the discrete curve at every point; and the parameters
$$
(\hat{z}_i, \hat{F}_i)\in \mathcal{I}:= \R^3\times SO(3),
$$
representing the pose of the particle, i.e. spatial location and orientation.

Recall moreover that, from Assumption \ref{assumption: known conformation},  we are given a reference atomic model, for which we can compute the parameters $(\Theta_0, \Psi_0, \hat{z}_0, \hat{F}_0)\in \mathcal{A}\times \mathcal{I}$.
We know that the structure of the macromolecule in every cryo-EM image is obtained after a deformation of the reference atomic model, which depends on the specific conformation, and a rigid transformation, which accounts for the pose of the particle.

With the above parametrization of the backbone, any rigid transformation of the reference structure can be obtained as
$$
\hat{z}_i = \hat{z}_0 + \Delta_{z,i} \qquad \hat{F}_i = \Delta_{F,i} \hat{F}_0
$$ 
for some $\Delta_{z,i} \in \R^3$ and $\Delta_{F,i} \in SO(3)$;
whereas any non-rigid deformation of the structure can be obtained as
$$
\Theta_i = \Theta_0 + \Delta_{\Theta,i} \qquad 
\Psi_i = \Psi_0 + \Delta_{\Psi,i},
$$
for some $\Delta_{\Theta,i}\in \R^{m-2}$ and 
$\Psi_i = \Psi_0 + \Delta_{\Psi,i} \in \R^{m-2}$.

As mentioned in the introduction, in this work we are interested in recovering the conformational heterogeneity of the macromolecule.
More precisely, we want to estimate the deformation of the given atomic model in each cryo-EM image, i.e. the increments $\Delta_{\Theta,i}$ and $\Delta_{\Psi,i}$ for each particle $i\in \{1, \ldots, n\}$.
Therefore we assume that the position and the orientation of the particles have been already estimated, for instance, by using existent methods such as RELION \cite{scheres2012relion}.

\begin{assumption}
\label{assumption: known directions}
For each 2D image in the cryo-EM dataset $\{ Y_i\}_{i=1}^n$, we know the 3D orientation and the position of the underlying particle, i.e.  we know the parameters $\hat{z}_i$ and $\hat{F}_i$ associated to each particle $u_i$.
\end{assumption}

As presented in subsection \ref{subsec: approx spectr decomp}, 
the parameters $\Theta_i \in [-\pi, \pi]^{m-2}$ and $\Psi \in [-\pi , \pi]^{m-2}$ for each cryo-EM image are estimated as
\begin{equation}
\label{spectral decomp reconstruction}
\Theta_i (A) = \Theta_0 + A \Phi_i \qquad
\text{and} \qquad \Psi_i (B) = \Psi_0 + B \Phi_i,
\end{equation}
where $K\in \mathbb{N}$ is a hyperparameter to be chosen (see Remark \ref{rmk: high ferquencies}),  representing the number of eigenvectors in the spectral decomposition; the vectors $\Phi_i\in \R^K$ for each $i$ are constructed by taking the $i$-th element in each of the $K$ leading eigenvectors of the graph Laplacian matrix, and the matrices $(A,B)\in [\R^{(m-2)\times K}]^2$ are the coefficients to be estimated.

The matrices $A$ and $B$ are estimated by means of an optimization problem in which the discrepancy between the predicted atomic models and the cryo-EM images is minimized.
In order to compare atomic models (atom positions) with 2D images,  we associate, to each atomic model, a 3D density obtained as a sum of Gaussian functions centred at the atom positions, and then we apply the cryo-EM forward operator \eqref{cryo-EM forward operator intro}.
The functional to be minimized consists on the mean square error between the obtained projections and the cryo-EM images.
Let us describe in detail how this optimization problem is formulated.

Using the expression \eqref{spectral decomp reconstruction} for the rotation angles, along with the known 3D spatial location $\hat{z}_i\in \R^3$ and  orientation $\hat{F}_i\in SO(3)$ of each particle, we can estimate the positions of the atoms in the backbone as the solution to the discrete dynamical system \eqref{dynamical system}.
For every particle $i\in \{1,,\ldots , n\}$, we denote the associated point cloud by
\begin{equation}
\label{Gamma_i def}
\Gamma_i (A,B) := \mathbf{Z}
(\Theta_i (A) ,  \Psi_i (B),  \hat{z}_i, \hat{F}_i) = \left[
z_1, z_2, \ldots z_m
\right] \in \R^{3m},
\end{equation}
where
$\mathbf{Z}: \mathcal{A} \times \mathcal{I} \longrightarrow
\R^{3m}$
is the operator which associates, to each set of parameters $(\Theta, \Psi, \hat{z}, \hat{F})\in \mathcal{A}\times \mathcal{I}$, the solution to the dynamical system \eqref{dynamical system}.
See appendix \ref{appdx: atomic model param} for more details.

Then, the 3D electrostatic potential of the underlying particle is estimated as a sum of Gaussian densities centred at the atom positions, i.e.
\begin{equation}
\label{density estimation}
\hat{U}_i (A,B) (x) =
 \sum_{z_j\in \Gamma_i(A,B)} \nu \exp
\left[-\dfrac{\| x - z_j\|^2 }{2\sigma^2} \right], \qquad \forall x\in \R^3,
\end{equation}
where $\nu, \sigma >0$ are hyperparameters which can be typically known from the chemical composition of the macromolecule.

Finally,  in order to compare the reconstructions with the cryo-EM images, we apply the forward cryo-EM operator to the  potential $\hat{U}_i (A,B) (\cdot)$ estimated for each particle,  and evaluate it in the 2D grid formed by the pixel positions of the images.
Let $\mathbf{X} := \{ x_{\alpha,\beta} \}_{\alpha,\beta = (1,1)}^{(N,N)} \in (\R^2)^{N\times N}$ be the rectangular grid corresponding to the pixel positions of the cryo-EM images.
Then we define $\hat{Y}_i (A,B) = \{\hat{y}_{\alpha, \beta}^{(i)}\}_{\alpha,\beta = (1,1)}^{(N,N)}$ as
\begin{equation}
\label{cryo-EM images estimate}
\hat{y}_{\alpha, \beta}^{(i)} (A,B) = 
h_i\ast T[\hat{U}_i (A,B)] (x_{\alpha, \beta}) \qquad
\text{for all} \ (\alpha, \beta) \in  \{ 1, 2, \ldots ,N  \}^{N\times N},
\end{equation}
where $T$ is the parallel beam ray transform taken along the vertical axis, followed by a 2D convolution with the point spread function $h_i$, that is assumed to be known for each image.   

The parameters $A$ and $B$ in \eqref{spectral decomp approx} are estimated by means of a minimisation problem, in which the loss functional is the mean square error between the clean cryo-EM images $\hat{Y}_i$ given by the estimated atomic model and the noisy cryo-EM images $Y_i$ from the dataset, i.e.
\begin{equation}
\label{variational problem}
[\hat{A}, \hat{B}] = \underset{A,B}{\operatorname{argmin}} \dfrac{1}{n} \sum_{i=1}^n  \| \hat{Y}_i(A,B) - Y_i\|_2^2.
\end{equation}

A solution to the minimisation problem \eqref{variational problem} can be approximated by means of stochastic gradient descent SGD algorithm in the following way.
Instead of computing the whole gradient at each iteration,  at the beginning of each epoch, the cryo-EM dataset is randomly partitioned in disjoint mini-batches
$$
\mathfrak{B}_1 \cup \mathfrak{B}_2 \cup \ldots \cup \mathfrak{B}_L = \{ 1, 2, \ldots , n\}
$$
with $L\in \mathbb{N}$ smaller than $n$.
Then, for each of the mini-batches  $\mathfrak{B}_l$ we update the parameters $(A,B)$ as
$$
(A,B) \longmapsto (A,B) - \gamma \dfrac{1}{n_l}  \nabla_{A,B} \left[ \sum_{i\in \mathfrak{B}_l}  \| \hat{Y}_i(A,B) - Y_i\|_2^2\right],
$$
where $\gamma >0$ is the learning rate and $n_l = | \mathfrak{B}_l|$ is the batch size.

\begin{remark}[On the non-convexity of problem \eqref{variational problem}]
We claim that the functional to be minimized in \eqref{variational problem} is non-convex with respect to $A$ and $B$. 
However, as we explain below in this remark, the specific form of the functions $\Theta(A)$ and $\Psi(B)$ in \eqref{spectral decomp reconstruction}, along with a suitable initialization (see Remark \ref{rmk: SGD initialization}),  can provide that the SGD algorithm converges to a global minimum.

In view of \eqref{spectral decomp reconstruction},  we are actually optimizing over the rotation angles $\Theta$ and $\Psi$ of the atomic structures.
The atom positions are given by the solution to the dynamical system \eqref{dynamical system}, where
$$
\Theta = (\theta_1, \theta_2 , \ldots , \theta_{m-2}) \in [-\pi , \pi]^{m-2}
\qquad \text{and} \qquad
\Psi = (\psi_1, \psi_2, \ldots , \psi_{m-2})\in [-\pi, \pi]^{m-2}.
$$
These parameters enter in \eqref{dynamical system} through the map
$$
(\theta_j , \psi_j) \longmapsto 
R (\theta_j, \psi_j) = 
\left(\begin{matrix}
\cos \psi_j \cos \theta_j & \cos \psi_j \sin \theta_j & - \sin \psi_j \\
-\sin \theta_j & \cos \theta_j & 0 \\
\sin \psi_j \cos \theta_j & \sin \psi_j \sin \theta_j & \cos \psi_j
\end{matrix}\right)\in SO(3),
$$
and therefore, we deduce that we are actually optimizing over the Lie group $SO(3)^{m-2}$.
This is a compact manifold with no boundary, which makes our minimization problem \eqref{variational problem} non-convex.
This may lead to problems when applying iterative algorithms such as SGD, as one can get stuck in a local minimum,  far away from any global minimizer.
However, we recall that we are looking for deformations of a given atomic model, i.e. perturbations of the known parameters $(\Theta_0, \Psi_0)$. 
This means that we are not optimizing over the whole $SO(3)^{m-2}$, but rather over a neighbourhood of the rotation matrices associated to the parameters $(\Theta_0, \Psi_0)$.
In view of the specific form of the functions $\Theta(A)$ and $\Psi(B)$, we look for matrices $A$ and $B$ close to zero.
\end{remark}

\begin{remark}[On the initialization]
\label{rmk: SGD initialization}
In this high-dimensional  non-convex situations,the success of iterative algorithms such as SGD heavily depends on having a good initialisation of the parameters. Here is where the specific form of the functions $\Theta_i(A)$ and $\Psi_i(B)$ in \eqref{spectral decomp reconstruction} becomes important, as we are looking for perturbations of $\Theta_0$ and $\Psi_0$.
Indeed,  by taking the initialization
$$
A_0 = \mathbf{0}_{m-2,K} \qquad \text{and} \qquad
B_0 = \mathbf{0}_{m-2,K}
$$ where $\mathbf{0}_{m-2,K}$ is the matrix of zeros of dimension $(m-2)\times K$, we obtain that 
$$
\Theta_i(A)=\Theta_0 \qquad \text{and} \qquad
\Psi_i(B) = \Psi_0 \qquad \forall i\in \{ 1,2,\ldots , n\}.
$$
In other words,  the predicted structure for all the particles is the given known conformation. 
Then, after each SGD iteration,  the coefficients $A$ and $B$ are updated, deforming the given structure according to the cryo-EM dataset.
\end{remark}

\section{Numerical experiments}
\label{sec: numerical experiments}

In this section, we present some numerical experiments in which we apply our method to carry out the tomographic reconstruction of flexible atomic structures in two synthetic datasets.  
In the first dataset, the underlying structure is a discrete curve in $\R^2$, which forms a flexible box and two moving arms (see figure \ref{fig: 2D given conformation}).  The volumetric representations of the structure are therefore functions in $\R^2$, and the associated tomographic projections are 1D functions (see figure \ref{fig: 2d struct data}). 
Although in the above, we have presented our method in a three-dimensional setting, it is not difficult to adapt all the steps in our method to a 2D structure with 1D tomographic projections.
In the second dataset, the underlying structure is a protein backbone, composed by the C-$\alpha$ atoms, which is deformed following a simulated molecular trajectory  from  \cite{Beckstein2018} using molecular dynamics (see Figure \ref{fig: backbone trajectory}).
In this case the structure is three-dimensional, and we can therefore apply our method exactly as it is presented in the above.

\subsection{Two-dimensional reconstruction}
\label{subsec: example 2d}

In this experiment, we consider a 2D structure consisting of a discrete curve in $\R^2$ with $m=149$ points and inter-atomic distance $\delta = 3.5$.
The atomic structure can be represented using the same construction as in Appendix \ref{appdx: atomic model param}, adapted to a two-dimensional setting:
\begin{equation}
\label{2D dynamical system}
\begin{cases}
z_{j+1} = z_j + \delta e_2 F_j & j \in \{ 1, \ldots , m-1\} \\
F_{j+1} = R(\theta_j) F_j & j\in \{ 1, \ldots , m-2\} \\
\text{with $z_{j_0} = \hat{z}\in \R^2$ and $F_{j_0} = \hat{F}\in SO(2)$,}
\end{cases}
\end{equation}
where $e_2 = (0,1)\in \R^2$, $\{z_j\}_{j=1}^m\in \R^{2m}$ denote the positions of the atoms and $\{ F_j\}_{j=1}^{m-1}\in SO(2)^{m-1}$ denote the discrete Frenet frames at each point in the curve.
Note that, since we have a 2D structure, the rotation matrices $R(\theta_j)\in SO(2)$ are parametrised by a single angle and are given by
$$
R (\theta_j) := \begin{bmatrix}
\cos \theta_j & -\sin \theta_j \\
\sin \theta_j & \cos \theta_j
\end{bmatrix}.
$$

\begin{figure}
\centering
\includegraphics[scale=.25]{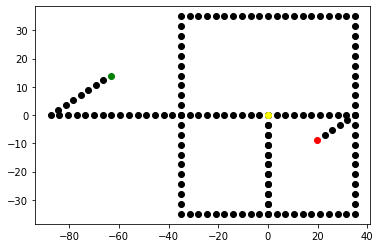}
\includegraphics[scale=.25]{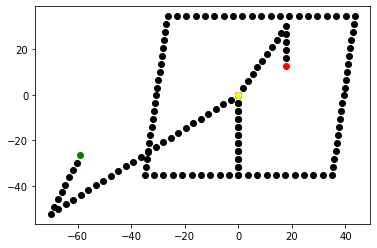}
\caption{2D structure in the given conformation (at the left) and a different conformation (at the right).  The first point $z_1$ is coloured in green,  the end point $z_m$ in red, and the reference point $z_{j_0}$ in yellow. Note that the curve self-intersects along the segment joining the box and the arms.}
\label{fig: 2D given conformation}
\end{figure}

In the synthetic dataset that we have generated, the 2D structure is a discrete curve in $\R^2$ forming a flexible box with two moving arms (see Figure \ref{fig: 2D given conformation}).
The arms and the box are joined by a segment, along which the atomic structure self-intersects.
We see that all the rotation angles  $\theta_j$ are equal to 0 except for the points in the corners of the structure.
In this experiment, the angles determining the position of the arms are given by
\begin{equation}
\label{arms angles}
\theta_{33} = -\frac{\pi}{2} + \tilde{\theta}_{33} \qquad
\text{and} \qquad 
\theta_{133} = -\frac{\pi}{2} + \tilde{\theta}_{133},
\end{equation}
where $ \tilde{\theta}_{33}$ and $ \tilde{\theta}_{133}$ are random variables uniformly distributed in $[-\pi/2, \pi/2]$,
and the angles determining the form of the box are given by
\begin{equation}
\label{box angles}
\theta_{53} = \theta_{93} =-\frac{\pi}{2} + \tilde{\theta}_{box} \quad
\text{and} \quad 
 \theta_{73} = \theta_{123} = -\frac{\pi}{2} - \tilde{\theta}_{box},
\end{equation}
where $\tilde{\theta}_{box}$ is a random variable uniformly distributed in $[-\pi/4, \pi/4]$.
As reference point we have chosen $j_0 = 33$,  which is in a part of the structure invariant in all the conformations (see Figure \ref{fig: 2D given conformation}).

In this two-dimensional case, the electrostatic potential associated to each particle is not a volumetric density, but rather a 2D density.
For the given point cloud $\{ z_j \}_{j=1}^m \in \R^{2m}$ we define the associated 2D density as
\begin{equation}
\label{image density 2d struct}
u (x) = \sum_{j=1}^m \exp \left[ -\dfrac{\| x-z_j\|^2}{2\sigma^2} \right], \qquad \forall x\in \R^2,
\end{equation}
with $\sigma = 7$.
In our numerical experiments we have considered no point spread function. The forward operator is modelled as the tomographic projection of the function $u(x)$ along the vertical axis, resulting in a 1D function.

In order to generate the synthetic cryo-EM dataset, we have generated $4000$ structures as defined above, where the point cloud is given by \eqref{2D dynamical system}, with the rotation angles given by \eqref{arms angles}--\eqref{box angles}.
Then, each structure is rotated by choosing $\hat{F}$ uniformly at random in $SO(2)$.
With the 4000 structures, we have generated the associated 2D densities $\{ u_i\}_{i=1}^{4000}$ as in \eqref{image density 2d struct}.
Then, we have computed the 1D tomographic projections by integrating the 2D densities along the vertical axis, i.e.
$$
Y_i = T(u_i) + \varsigma_i, \qquad \text{for} \ i = 1, \ldots , 4000,
$$
where $T(\cdot)$ denotes the tomographic projection of $u_i$ along the vertical axis, evaluated on a one-dimensional grid with 128 discretisation points.
We have also added Gaussian noise  $\varsigma_i\in \R^{128}$, where each component follows a Gaussian distribution with variance $2500$.
The variance of the clean projections is $182.63$, so the signal to noise ratio of the dataset is
$$
SNR = \dfrac{182.63}{2500} = 0.073.
$$

\begin{figure}
\centering
\begin{subfigure}[c]{.25\textwidth}
\centering
\includegraphics[scale=.25]{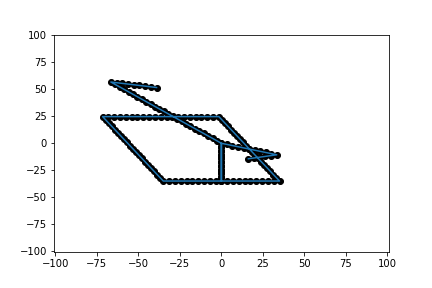}
\includegraphics[scale=.25]{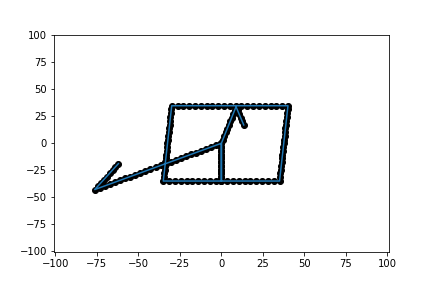}
\end{subfigure}
\hfill
\begin{subfigure}[c]{.2\textwidth}
\includegraphics[scale=.28]{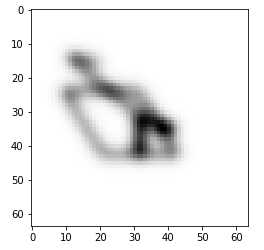}
\includegraphics[scale=.28]{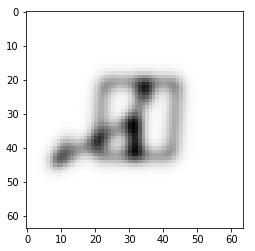}
\end{subfigure}
\hfill
\begin{subfigure}[c]{.25\textwidth}
\includegraphics[scale=.25]{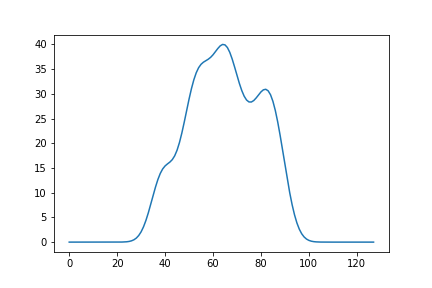}
\includegraphics[scale=.25]{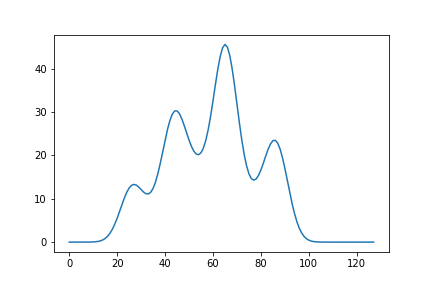}
\end{subfigure}
\hfill
\begin{subfigure}[c]{.25\textwidth}
\includegraphics[scale=.25]{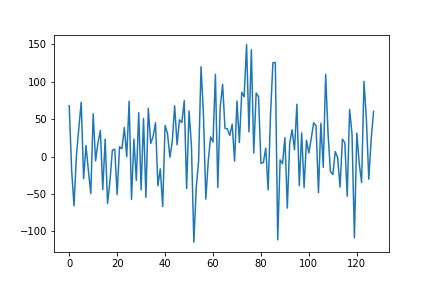}
\includegraphics[scale=.25]{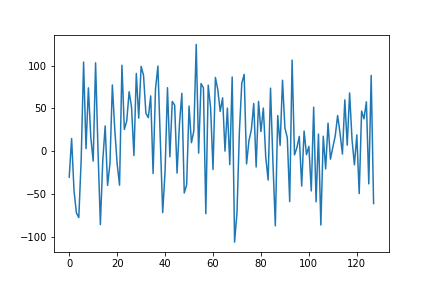}
\end{subfigure}
\caption{From left to right: the 2D structure, the 2D density function, the clean tomographic projection and the noisy tomographic projection.}
\label{fig: 2d struct data}
\end{figure}

As mentioned in Assumption \ref{assumption: low-dimension reconstruction}, we assume that we have a low-dimension reconstruction of the particles in the cryo-EM dataset (this one can be obtained by the method in \cite{anden2018structural}).
In our case, this low-dimension representation of the particles comes in the form of noisy 2D densities of size  $64\times 64$, that we call $\{ U_i\}_{i=1}^{4000} \subset \R^{64\times 64}$. Each density $U_i$ is defined as the function \eqref{image density 2d struct} associated to the $i$-th point cloud, evaluated in a 2D grid. We then added Gaussian noise with variance 9 to each low-resolution density.
With these densities we have constructed the weighted graph of similarities, using a Gaussian kernel of the form
\begin{equation}
\label{graph weights numerical exper}
W_{ij} = \exp \left[- \dfrac{\| U_i - U_j \|^2}{2\sigma^2}  \right], \quad \text{with} \ \sigma = 96.
\end{equation}
Then we have constructed the normalised graph Laplacian and computed the eigenvectors associated to the 20 smallest eigenvalues.
In Figure \ref{fig: eigenvec 2d} we see the representation of some of these eigenvectors. These plots can be seen as an approximation of the manifold of conformations $\mathcal{M}$. Note that this structure has three degrees of freedom (the movement of the two arms plus the deformation of the box), and therefore, the manifold of conformations is expected to have dimension $3$.

The next step in our method is to construct the spectral decomposition of the rotation angles $\theta_j$ with $j=1, \ldots, m-2$.
Note that in the  2D case we only need to estimate one rotation angle $\theta_j$ at each point in the curve, as opposed to two angles in the case of 3D structures.
Let us denote the sequence of rotation angles for the $i$-th particle as 
$$
\Theta_i = [\theta_1, \theta_2, \ldots , \theta_{m-2}]\in [-\pi, \pi]^{m-2}.
$$
As aforementioned, we assume to have access to a conformation of the structure (see Assumption \ref{assumption: known conformation}).
In this case, the given conformation (see figure \ref{fig: 2D given conformation}) is the structure with the arms horizontally aligned and the box forming a square, i.e.  we take
$\tilde{\theta}_{33} = \tilde{\theta}_{133} = \tilde{\theta}_{box} = 0$ in \eqref{arms angles} and \eqref{box angles}.
Following subsection \ref{subsec: approx spectr decomp}, the parameter $\Theta_i$ for each particle is estimated as
$$
\Theta_i = \Theta_0 + A \Phi_i,
$$
where $\Phi_i\in \R^{20}$ is the vector with the $i$-th component of each of the first eigenvectors of the graph Laplacian computed before.

The last step in our method is to estimate the matrix $A \in \R^{(m-2)\times 20}$, which we address by following subsection \ref{subsec: tomographic reconstruction}.
We will consider two cases for the tomographic reconstruction.
In the first one, we use the knowledge of the angles that are actually varying, i.e. \eqref{arms angles} and \eqref{box angles}.  Note that, by the construction of this particular structure, most of the rotation angles are equal to 0 in all the conformations.
In the second case, we do not assume such knowledge about the angles which are actually varying, and therefore, the entire vector $\Theta_i$ is estimated.
From a practical viewpoint a similar assumption may hold if one knows which  parts in the macromolecule are flexible and which are not. For instance, it is known that the secondary structures in proteins, such as $\alpha$-helices and $\beta$-strands are preserved from one conformation to another (see Remark \ref{rmk: high ferquencies}).

The coefficients in the matrix $A$ are estimated by applying SGD to the minimisation problem \eqref{variational problem}. In the minimisation problem, we use only 90 percent of the data ($n=3600$ noisy tomographic projections).
The rest of the data ($n=400$ clean tomographic projections)  are used to test the accuracy of the reconstructed structures. 
In both cases, we have applied 100 epochs of SGD algorithm with a batch size of 500.
As for the learning rate, in the case in which only the varying parameters are estimated (subsection \ref{subsubsec_ with prior}) we used a learning rate of 0.1, whereas in the case in which all the parameters are estimated (subsection \ref{subsubsec_ no prior}) it seems to be necessary to take a smaller learning rate, and we have chosen 0.01 for this case.
In both cases, we used the PyTorch implementation of the SGD algorithm, and it took around 2 minutes to compute 100 epochs of SGD on a MacBook Pro 1.4 GHz Intel Core i5 with 16 GB of RAM memory.

\subsubsection{Reconstruction estimating only the varying angles}
\label{subsubsec_ with prior}

In this case, we assume that we know the angles that are actually varying in the structure, i.e. in view of \eqref{arms angles} and \eqref{box angles}, we only need to estimate the components $\theta_{33}, \theta_{53}, \theta_{73}, \theta_{93}, \theta_{123}$ and $\theta_{133}$ in each vector $\Theta_i$.
All the other parameters in $\Theta_i$ are given by the known conformation $\Theta_0$.
In Figure \ref{fig: training 2d loss}, at the left, we observe the evolution of the loss function evaluated after every epoch of the SGD algorithm. In the centre we see the evolution of the loss function evaluated on the test data. Note that the test data consists of clean tomographic projections, and therefore, the loss is much smaller. At the right we see the evolution of the  distance between the predicted point clouds for the test data $\{\tilde{z}_i\}_{i=1}^{400}$ and the ground truths  $\{z_i\}_{i=1}^{400}$. The orange line in figure \ref{fig: training 2d loss} represents the maximum over the  point cloud
\begin{equation}
\label{max point cloud error}
\dfrac{1}{400} \sum_{i=1}^{400} \max_{j = 1, \ldots , m} \| z_{i,j} - \tilde{z}_{i,j} \|,
\end{equation}
whereas the blue line represents the total average error in the predicted point cloud with respect to the ground truth, i.e.
\begin{equation}
\label{avg point cloud error}
\dfrac{1}{400m} \sum_{i=1}^{400} \sum_{j = 1}^m \| z_{i,j} - \tilde{z}_{i,j} \|.
\end{equation}
In Figure \ref{fig: 2d reconstruction} we see some of the reconstructions of the 2D structures.
Of course,  using the knowledge of which parameters are varying and which are invariant over the conformations yields significantly better reconstructions as compared to the reconstructions obtained in the following subsection \ref{subsubsec_ no prior}, where such a knowledge is not assumed.

\begin{figure}[t]
\centering
\includegraphics[scale=.3]{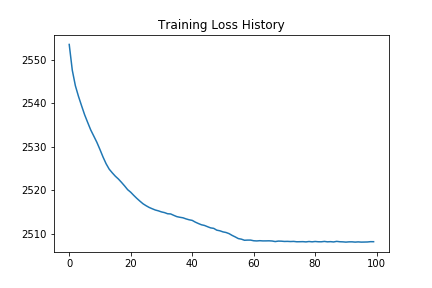}
\includegraphics[scale=.3]{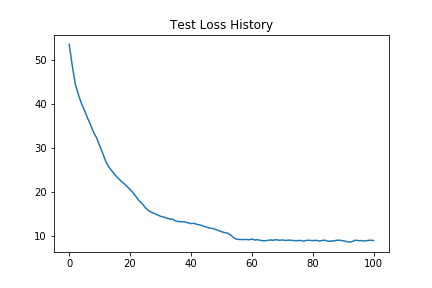}
\includegraphics[scale=.3]{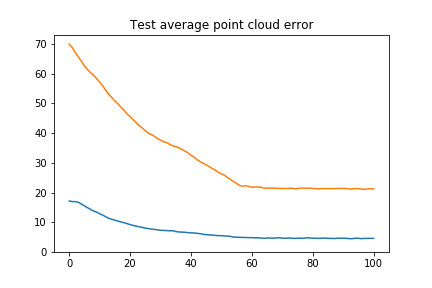}
\caption{Evolution of the loss functional and the error in the predictions with respect to the ground truth for the 2D reconstructions, using the prior knowledge about the varying parameters (subsection \ref{subsubsec_ with prior}).}
\label{fig: training 2d loss}
\end{figure}

\begin{figure}[t]
\centering
\begin{subfigure}[c]{.48\textwidth}
\centering
\includegraphics[scale=.5]{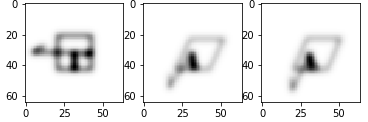}
\includegraphics[scale=.5]{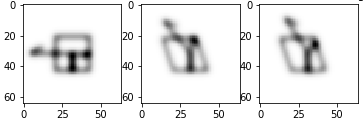}
\includegraphics[scale=.5]{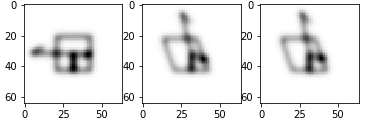}
\caption{Reconstruction of the 2D densities using our method. At the left we see the structure in the given conformation, in the middle the ground true and at the right the reconstruction.}
\label{fig: 2d reconstruction a}
\end{subfigure}
\hfill
\begin{subfigure}[c]{.48\textwidth}
\centering
\includegraphics[scale=.25]{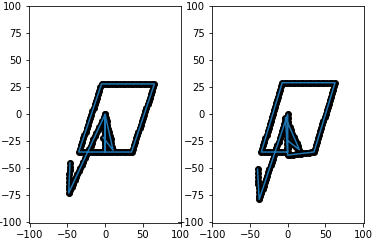}\\
\includegraphics[scale=.25]{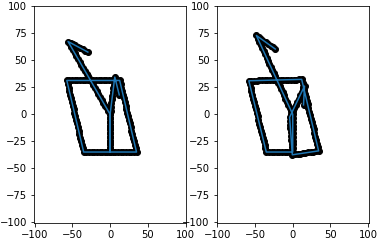}\\
\includegraphics[scale=.25]{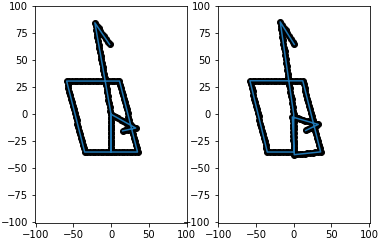}
\caption{At the left the 2D structure (ground truth) and at the right the structure predicted by means of our method.}
\label{fig: 2d reconstruction b}
\end{subfigure}
\caption{Some examples of the tomographic reconstruction of 2D structures from tomographic projections. In this case we used the knowledge of the parameters that are varying.}
\label{fig: 2d reconstruction}
\end{figure}

\subsubsection{Reconstruction estimating all the angles}
\label{subsubsec_ no prior}

As opposite to the previous case, we now do not assume the knowledge of the varying parameters in the structure, and therefore we estimate the spectral decomposition of all the rotation angles  in $\Theta$.
In Figure \ref{fig: training 2d loss no prior} we see the evolution of the loss functional over the epochs of the SGD algorithm as well as the error in the predicted point clouds in the test dataset (we recall that the orange line is given by \eqref{max point cloud error} and the blue line is given by \eqref{avg point cloud error} in the previous subsection).
We observe in Figure \ref{fig: 2d reconstruction no prior} that the predicted reconstructions are less accurate as compared with the reconstructions obtained when we know the parameters that are varying.

\begin{figure}[t]
\centering
\includegraphics[scale=.3]{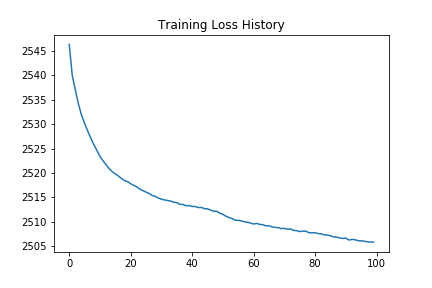}
\includegraphics[scale=.3]{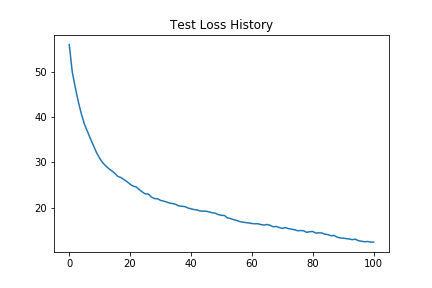}
\includegraphics[scale=.3]{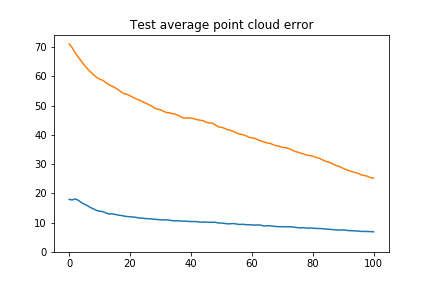}
\caption{Evolution of the loss functional and the error in the predictions with respect to the ground truth for the 2D reconstructions, without using the prior knowledge about the varying parameters (subsection \ref{subsubsec_ no prior}).}
\label{fig: training 2d loss no prior}
\end{figure}

\begin{figure}[t]
\centering
\begin{subfigure}[c]{.48\textwidth}
\centering
\includegraphics[scale=.5]{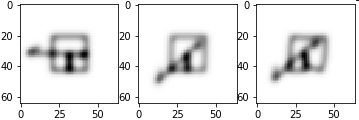}
\includegraphics[scale=.5]{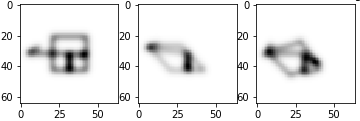}
\includegraphics[scale=.5]{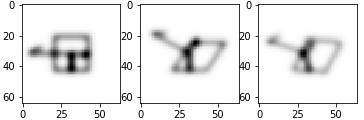}
\caption{Reconstruction of the 2D densities using our method. At the left we see the structure in the given conformation, in the middle the ground true and at the right the reconstruction.}
\end{subfigure}
\hfill
\begin{subfigure}[c]{.48\textwidth}
\centering
\includegraphics[scale=.25]{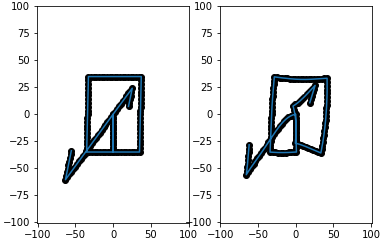}\\
\includegraphics[scale=.25]{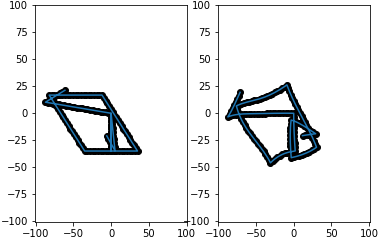}\\
\includegraphics[scale=.25]{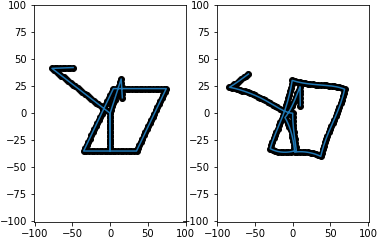}
\caption{At the left the 2D structure (ground truth) and at the right the structure predicted by means of our method.}
\end{subfigure}
\caption{Some examples of the tomographic reconstruction of 2D structures from tomographic projections. In this case we did not use the knowledge of the parameters that are varying, and therefore, all the parameters are estimated.}
\label{fig: 2d reconstruction no prior}
\end{figure}

\subsection{Three-dimensional reconstruction}
\label{subsec: example3d}

In this subsection, we consider two 3D atomic structures satisfying the chain property from Assumption \ref{assumption: discrete curve}. 
The structures correspond to proteins from which we have extracted the positions of the C-$\alpha$ atoms.   
In both cases the ground truth structures for the different conformations are simulated trajectories using molecular dynamics, taken from \cite{Beckstein2018} and \cite{shaw2020molecular} respectively.
The trajectories consist of a sequence of atomic models corresponding to the different frames of the molecular trajectory.

In order to generate a synthetic cryo-EM dataset from the MD simulation, we have chosen 4000 structures uniformly at random from the frames of the MD trajectory. Then, each particle is rotated by randomly selecting a   rotation matrix $\hat{F}_i \in SO(3)$ from a uniform distribution.
The cryo-EM images $\{ Y_i\}_{i=1}^{4000}$ are obtained as 
$$
Y_i = h_i\ast T(u_i) + \varsigma_i,  \qquad \text{for}  \  i =  1, \ldots , 4000
$$
where $T(\cdot)$ is the parallel beam ray transform of the 3D density volume $u_i(\cdot)$ along the vertical axis, evaluated in a $64\times 64$ two-dimensional grid.

Since the 2D convolution is applied to the images in Fourier variables, we simply multiply the 2D projection by the CTF $\hat{h}_i$, as defined in \eqref{CTF intro}. 
In all the numerical experiments we have used the same parameters for the CTF.  The voltage of the transmission electron microscope is set to $200$ keV, the spherical aberration $C_s$ is of $2$ mm,  the amplitude contrast is set to $\alpha = 0.1$, and the aperture function $A(\xi)$ is equal to $1$. As for the defocus,  it is typical in practice to have images at different defocus in order to avoid the zero-crossings of the CTF to coincide in all the images. Therefore, in each cryo-EM image, we have selected a random defocus $\Delta z$ among the values $(1.5, 1.75, 2, 2.25, 2.5)$ in microns.

The 3D densities $\{ u_i\}_{i=1}^{4000}$ corresponding to each particle are generated as the sum of Gaussian functions centred at the C-$\alpha$ positions as in \eqref{density estimation}, i.e.
\begin{equation}
\label{image density 3d struct}
u_i (x) = \sum_{j=1}^{n} \exp \left[ -\dfrac{\| x-z_{i,j}\|^2}{2\sigma^2} \right], \qquad \forall x\in \R^3.
\end{equation}
where $\{ z_{i,j} \}_{j=1}^{n}\in \R^{3\times n}$ are the positions of the C-$\alpha$ atoms in the $i-$th particle.
In both cases, we have taken $\sigma = 3$.
Finally, the images are corrupted with Gaussian noise,  resulting in SNR of the order $0.01$.
We have carried out the experiments at different levels of SNR, in order to compare the sensitivity of the 3D reconstructions to the noise.

Just as in the numerical experiments in subsection \ref{subsec: example 2d}, we assume that we have  a low-dimension representation of the 3D volumes associated to the particles (Assumption \ref{assumption: low-dimension reconstruction}).
In this case, these low-dimension representations are the density volumes \eqref{image density 3d struct} associated to each particle, evaluated in a 3D voxel grid of size $16\times 16 \times 16$.
In the experiments presented in subsection \ref{subsubsec: example 3d 2}, we have also run the experiment using higher resolution representations to construct the graph Laplacian, but the results seem to be similar.
In order to simulate possible inaccuracies in the low-dimension representation, we have added different levels of noise to the low-resolution volumes (see Tables \ref{talbe: expermients 1} and \ref{talbe: expermients 2}).
With these low-dimension representation of the 3D volumes $\{ U_i\}_{i=1}^{4000} \subset \R^{16\times 16\times 16}$, we construct a weighted graph with weights given by \eqref{graph weights numerical exper},  with $\sigma$ proportional to the noise.
Then, we have constructed the normalised graph Laplacian matrix associated to the above graph, and computed the eigenvectors associated to the 10 smallest eigenvalues.
In the molecular trajectory, we also observe high-frequency vibrations of the atoms, however, by keeping the number of eigenvectors rather small in the spectral decomposition, we are not recovering these vibrations (see Remark \ref{rmk: high ferquencies}).

Following subsection \ref{subsec: approx spectr decomp},  we construct the rotation angles for each particle as a linear combination of the 10 first eigenvectors, i.e.
$$
\Theta_i = \Theta_0 +  A \Phi_i
\qquad \text{and} \qquad
\Psi_i = \Psi_0 +  B \Phi_i,
$$
where $\Theta_0$ and $\Psi_0$ are the rotation angles associated to the known conformation (see Assumption \ref{assumption: known conformation}). In both experiments, we have chosen the known conformation $(\Theta_0,\Psi_0)$ to be the that of the particle in the first frame of the molecular trajectory (the atomic structure at the left in Figures \ref{fig: backbone trajectory} and \ref{fig: backbone trajectory 2}).

In the tomographic reconstruction we have not assumed any prior knowledge about the flexibility of the different parts of the atomic structure nor the secondary structures, and therefore, all the parameters in $\Theta$ and $\Psi$ are  estimated in the same way.  
The coefficients in the spectral decomposition of the rotation angles $\Theta$ and $\Psi$ are estimated by means of SGD applied to the minimisation problem \eqref{variational problem} in subsection \ref{subsec: tomographic reconstruction}.
As before, we have used only 90 percent of the cryo-EM data (3600 noisy tomographic projections) in the minimisation problem, leaving the remaining 10 percent (400 clean tomographic projections) to evaluate the accuracy of the reconstruction.

\subsubsection{Three-dimensional structure: Adenylate Kinase}
\label{subsubsec: example 3d 1}

The structure corresponds to the backbone of a protein with 214 amino-acid residues, from which we have extracted the positions of the C-$\alpha$ atoms. 
The motion of the protein structure is a simulation taken from \cite{Beckstein2018} using molecular dynamics,  and the result is a video with 102 frames in which the backbone is continuously deformed (see figure \ref{fig: backbone trajectory}).
See figure \ref{fig: 3d struct data CTF} for three examples of the synthetic cryo-EM images that we have generated for this experiment.

In the whole trajectory, the distance between adjacent C-$\alpha$ atoms in the backbone is approximately constant,  with average $3.8413$.
In the reconstruction of the atomic structure, we have used a constant inter-atomic distance $\delta = 3.8412$.
This choice produces small errors in the reconstructions as the inter-atomic distances in the dataset (ground truth) slightly vary due to small vibrations of the atoms.
However, these errors do not seem to be very relevant in our reconstruction.
As reference point in our reconstructions we have chosen the middle point in the discrete curve, i.e.  the point $j_0 = 112$.

\begin{figure}[t]
\centering
\includegraphics[scale=.4]{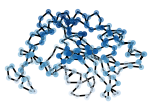}
\hspace{20pt}
\includegraphics[scale=.4]{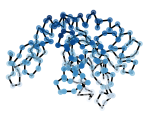}
\hspace{20pt}
\includegraphics[scale=.4]{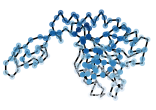}
\hspace{20pt}
\includegraphics[scale=.4]{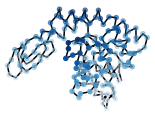}
\caption{The backbone of the protein in four frames of the MD trajectory used in the numerical experiments in subsection \ref{subsubsec: example 3d 1}.}
\label{fig: backbone trajectory}
\end{figure}

\begin{figure}[t]
\centering
\includegraphics[scale=.5]{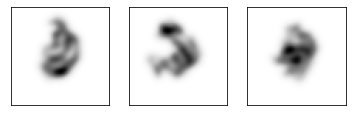}

\centering
\includegraphics[scale=.5]{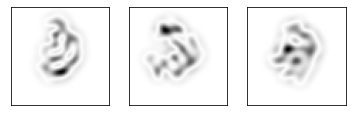}

\centering
\includegraphics[scale=.5]{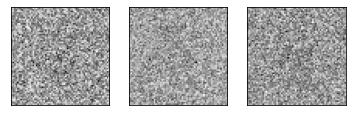}
\caption{Three cryo-EM images from the synthetic dataset used in the experiment corresponding to the 4-th row in Table \ref{talbe: expermients 1}. In the top row we see the clean tomographic projections of three particles, in the middle row we see these tomographic projections after the convolution with the point spread function, and in the bottom row we see the images after adding Gaussian noise. The SNR is $0.0136$.}
\label{fig: 3d struct data CTF}
\end{figure}

\begin{figure}
\centering
\begin{subfigure}{.24\textwidth}
\centering
\includegraphics[scale=.3]{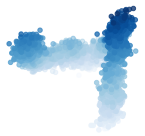}
\caption{$[\phi^{(1)}, \phi^{(2)}, \phi^{(3)}]$}
\end{subfigure}
\hfill
\begin{subfigure}{.24\textwidth}
\centering
\includegraphics[scale=.3]{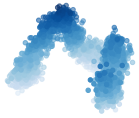}
\caption{$[\phi^{(1)}, \phi^{(2)}, \phi^{(4)}]$}
\end{subfigure}
\hfill
\begin{subfigure}{.24\textwidth}
\centering
\includegraphics[scale=.3]{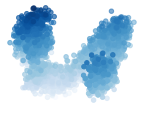}
\caption{$[\phi^{(1)}, \phi^{(3)}, \phi^{(4)}]$}
\end{subfigure}
\hfill
\begin{subfigure}{.24\textwidth}
\centering
\includegraphics[scale=.3]{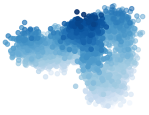}
\caption{$[\phi^{(2)}, \phi^{(3)}, \phi^{(4)}]$}
\end{subfigure}
\caption{Representation of the eigenvectors $\phi^{k}\in \R^n$ of the graph Laplacian constructed from the synthetic dataset used in subsection \ref{subsubsec: example 3d 1}.  Each plot can be seen as an approximation of the manifold of conformations $\mathcal{M}$ projected in different three-dimensional linear spaces.}
\label{fig: eigenvec 3d}
\end{figure}

In Table \ref{talbe: expermients 1} we see the accuracy of the numerical experiments carried out for this 3D structure at different noise levels. We observe that the results are very similar.
For each experiment, we have computed 30 iterations of SGD with a batch size of 100 and a learning rate of $0.5$.
In figure \ref{fig: training 3d loss} we see the evolution of the loss functional over the iterations of the SGD algorithm. In the plot at the right, we see the evolution of the error of our reconstructed point clouds with respect to the ground truth (see formulae \eqref{max point cloud error} and \eqref{avg point cloud error}). These plots correspond to the fourth row in Table \ref{talbe: expermients 1}.
We have used the PyTorch implementation of the SGD algorithm. This time,  it took around 30 minutes to compute 30 epochs on an Intel Core(TM) i7-7700 CPU at 3.60GHz with 32 GB of RAM memory.
See figure \ref{fig: 3d reconstruction} for some example of the tomographic reconstructions of the 3D structures presented in this subsection.

\begin{table}
\begin{tabular}{c|c||c|c|c|c}
SNR & Low-dim.  noise & Training loss & Test loss & Avg. error & Max. error \\
\hline
  0.0559 &0.5 & 9.1166 & 0.1121 & 2.4346 & 6.4855 \\
  0.0545 & 1 & 9.1072 & 0.1135 & 2.4887 & 6.7112   \\
  0.0135 & 0.5 & 36.1290 & 0.1110 & 2.4812 & 6.5343 \\
  0.0136 & 1 & 36.1282 & 0.1118 & 2.4718 & 6.6462
\end{tabular}
\caption{Results of the numerical experiments in subsection \ref{subsubsec: example 3d 1} at different noise levels.
SNR denotes the ratio between the variance of the clean tomographic projection and the cryo-EM images used in the reconstruction.  Low-dim. noise denotes the variance of the Gaussian noise added to the low-dimension representation of the particles used to construct the graph Laplacian. Avg. error denotes the average error between the predicted atom position in the atomic model and the ground truth (see formula \eqref{avg point cloud error}), and Max. error is the maximum distance between the predicted atom positions and the ground truth (see formula \eqref{max point cloud error}).}
\label{talbe: expermients 1}
\end{table}

\begin{figure}[t]
\centering
\includegraphics[scale=.3]{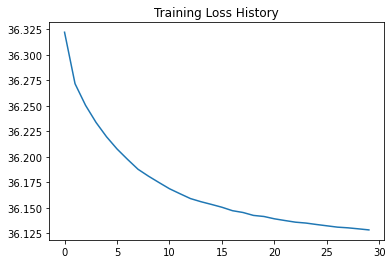}
\includegraphics[scale=.3]{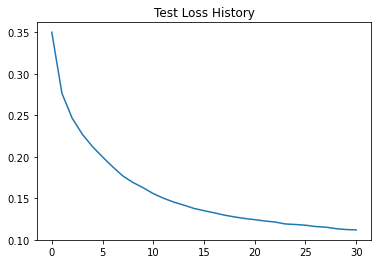}
\includegraphics[scale=.3]{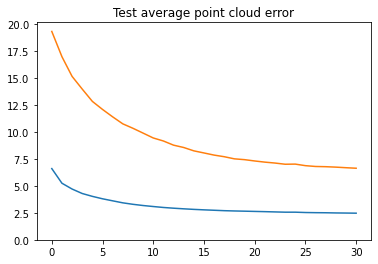}
\caption{Evolution of the loss functional and the error in the predictions with respect to the ground truth for the 3D reconstructions in subsection \ref{subsubsec: example 3d 1}.  The experiment corresponds to the 4-th row in Table \ref{talbe: expermients 1}.}
\label{fig: training 3d loss}
\end{figure}

\begin{figure}[t]
\centering
\begin{subfigure}[c]{.48\textwidth}
\centering
\includegraphics[scale=.4]{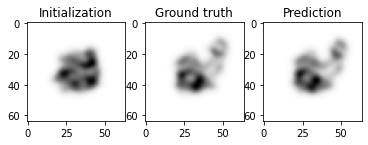}
\includegraphics[scale=.4]{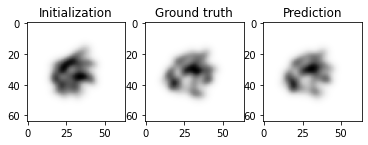}
\includegraphics[scale=.4]{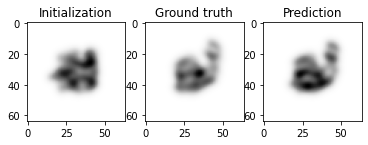}
\caption{Clean 2D tomographic projections of the reconstructed 3D structures. At the left we see the structure in the given conformation, in the middle the ground true and at the right the reconstruction.}
\end{subfigure}
\hfill
\begin{subfigure}[c]{.48\textwidth}
\centering
\includegraphics[scale=.3]{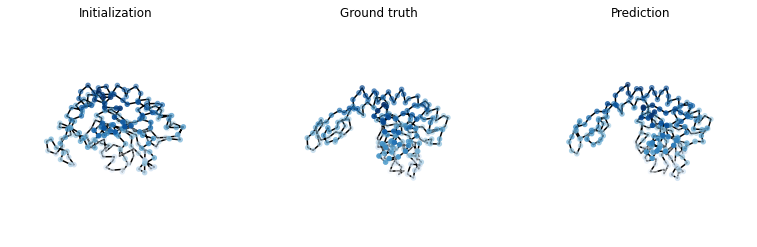}\\
\includegraphics[scale=.3]{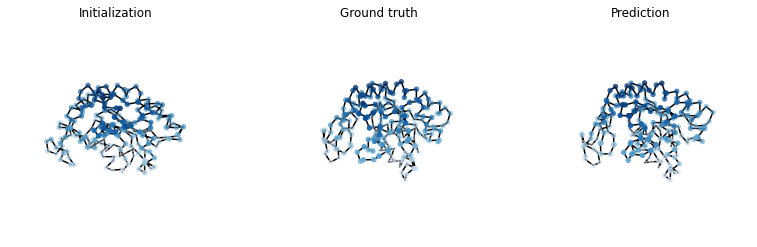}\\
\includegraphics[scale=.3]{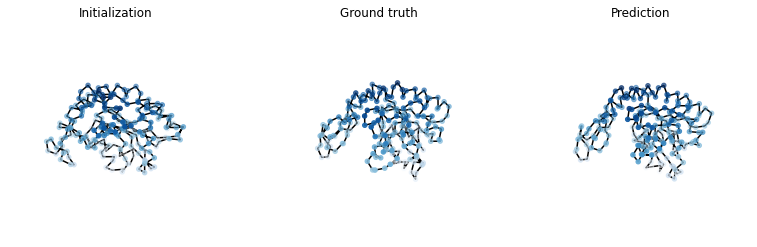}
\caption{At the left we see the structure in the known conformation, in the middle we see the particle in its specific conformation (ground truth) and at the right the 3D atomic structure predicted by means of our method.}
\end{subfigure}
\caption{Some examples of the tomographic reconstruction in subsection \ref{subsubsec: example 3d 1}. The reconstructions are taken from the experiment with noise levels as in 4-th row in Table \ref{talbe: expermients 1}.}
\label{fig: 3d reconstruction}
\end{figure}

\subsubsection{Three-dimensional structure: SARS-CoV-2 helicase nsp 13}
\label{subsubsec: example 3d 2}

The structure corresponds to the backbone of the molecule SARS-CoV-2 helicase nsp 13.  This is a chain structure with 590 C-$\alpha$ atoms, a much bigger macromolecule than the one in the previous subsection. 
The different conformations are taken from a trajectory simulated in \cite{shaw2020molecular} using molecular dynamics.  We have used 400 frames in which the backbone is continuously deformed (see figure \ref{fig: backbone trajectory 2}).
See figure \ref{fig: 3d struct data CTF 2} for three examples of the synthetic cryo-EM images that we have used in our experiments.

In the whole trajectory, the distance between adjacent C-$\alpha$ atoms in the backbone is approximately constant,  with average $3.85$ Angstroms.
In the reconstruction of the atomic structure, we have used a constant inter-atomic distance $\delta = 3.85$.
This choice produces small errors in the reconstructions as the inter-atomic distances in the dataset (ground truth) slightly vary due to small vibrations of the atoms.
However, these errors do not seem to be very relevant in our reconstruction (see Figure \ref{fig: 3d reconstruction 2}).

\begin{figure}[t]
\centering
\includegraphics[scale=.4]{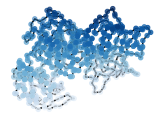}
\hspace{20pt}
\includegraphics[scale=.4]{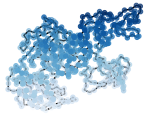}
\hspace{20pt}
\includegraphics[scale=.4]{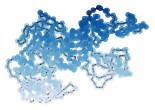}
\hspace{20pt}
\includegraphics[scale=.4]{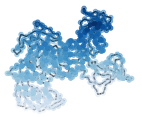}
\caption{The backbone of the macromolecule in four frames of the MD trajectory used in the numerical experiments in subsection \ref{subsubsec: example 3d 2}.}
\label{fig: backbone trajectory 2}
\end{figure}

\begin{figure}[t]
\centering
\includegraphics[scale=.5]{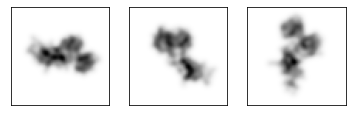}

\centering
\includegraphics[scale=.5]{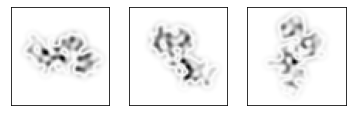}

\centering
\includegraphics[scale=.5]{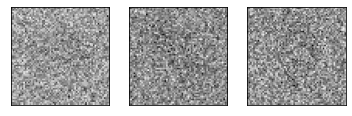}
\caption{Three cryo-EM images from the dataset used in the experiment corresponding to the 4-th row in Table \ref{talbe: expermients 2}. In the top row we see the clean tomographic projections of three particles, in the middle row we see the tomographic projections after the convolution with the point spread function, and in the bottom row we see the images after adding Gaussian noise. The SNR is $0.019$.}
\label{fig: 3d struct data CTF 2}
\end{figure}

\begin{figure}
\centering
\begin{subfigure}{.24\textwidth}
\centering
\includegraphics[scale=.3]{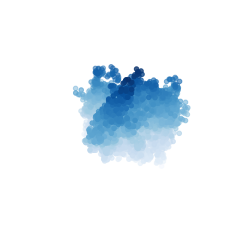}
\caption{$[\phi^{(1)}, \phi^{(2)}, \phi^{(3)}]$}
\end{subfigure}
\hfill
\begin{subfigure}{.24\textwidth}
\centering
\includegraphics[scale=.3]{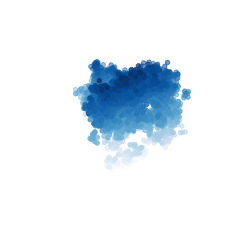}
\caption{$[\phi^{(1)}, \phi^{(2)}, \phi^{(4)}]$}
\end{subfigure}
\hfill
\begin{subfigure}{.24\textwidth}
\centering
\includegraphics[scale=.3]{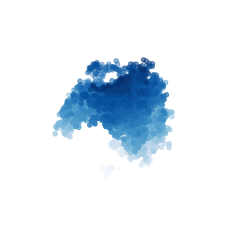}
\caption{$[\phi^{(1)}, \phi^{(3)}, \phi^{(4)}]$}
\end{subfigure}
\hfill
\begin{subfigure}{.24\textwidth}
\centering
\includegraphics[scale=.3]{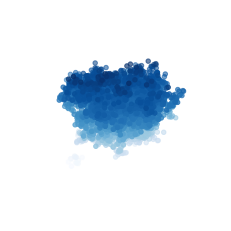}
\caption{$[\phi^{(2)}, \phi^{(3)}, \phi^{(4)}]$}
\end{subfigure}
\caption{Representation of the eigenvectors $\phi^{k}\in \R^n$ of the graph Laplacian constructed from the synthetic dataset of 3D structures in subsection \ref{subsubsec: example 3d 2}.  Each plot can be seen as an approximation of the manifold of conformations $\mathcal{M}$ projected in different three-dimensional linear spaces.}
\label{fig: eigenvec 3d 2}
\end{figure}

In Table \ref{talbe: expermients 2} we see the accuracy of the numerical experiments carried out for this 3D structure at different noise levels. 
For each experiment, we have computed 20 iterations of SGD with a batch size of 100 and a learning rate of $0.5$.
In figure \ref{fig: training 3d loss 2} we see the evolution of the loss functional over the iterations of the SGD algorithm. In the plot at the right, we see the evolution of the error of our reconstructed point clouds with respect to the ground truth (see formulae \eqref{max point cloud error} and \eqref{avg point cloud error}). These plots correspond to the fourth row in Table \ref{talbe: expermients 2}.
As before, we have used the PyTorch implementation of the SGD algorithm. This time,  it took around 80 minutes to compute 20 epochs on an Intel Core(TM) i7-7700 CPU at 3.60GHz with 32 GB of RAM memory.
See figure \ref{fig: 3d reconstruction} for some example of the tomographic reconstructions of the 3D structures presented in this subsection.

\begin{table}
\begin{tabular}{c|c|c||c|c|c|c}
SNR & Low-dim.  noise & 3Dvol px & Training loss & Test loss & Avg. error & Max. error \\
\hline
  0.0746 & 0.5 & 16 & 9.1768 & 0.1813 & 3.1334 & 14.6131 \\
  0.0786 & 0.5 & 32 & 9.1938 & 0.1860 & 3.1044 & 14.5323 \\
  0.0783 & 1 & 32 & 9.1844 & 0.1898 & 3.1290 & 15.7472 \\
  0.019 & 0.5 & 16 & 36.2054 & 0.2012 & 3.3200 & 15.6157 \\
\end{tabular}
\caption{Results of the numerical experiments in subsection \ref{subsubsec: example 3d 2} at different noise levels.
SNR denotes the ratio between the variance of the clean tomographic projection and the cryo-EM images used in the reconstruction.  Low-dim. noise denotes the variance of the Gaussian noise added to the low-dimension representation of the particles used to construct the graph Laplacian. 
We have tested if the results would improve by using higher resolution volumes in the construction of the graph, but it doesn't seem to have a great effect. 
Avg. error denotes the average error between the predicted atom position in the atomic model and the ground truth (see formula \eqref{avg point cloud error}), and Max. error is the maximum distance between the predicted atom positions and the ground truth (see formula \eqref{max point cloud error}).}
\label{talbe: expermients 2}
\end{table}

\begin{figure}[t]
\centering
\includegraphics[scale=.3]{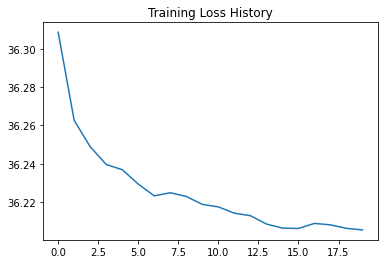}
\includegraphics[scale=.3]{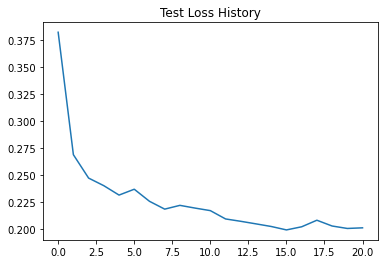}
\includegraphics[scale=.3]{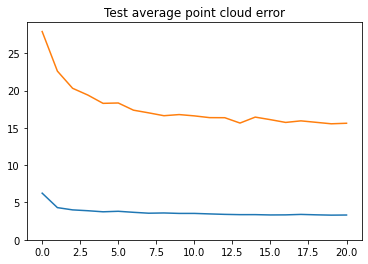}
\caption{Evolution of the loss functional and the error in the predictions with respect to the ground truth for the 3D reconstructions in subsection \ref{subsubsec: example 3d 2}.  The experiment corresponds to the 4-th row in Table \ref{talbe: expermients 2}.}
\label{fig: training 3d loss 2}
\end{figure}

\begin{figure}[t]
\centering
\begin{subfigure}[c]{.48\textwidth}
\centering
\includegraphics[scale=.4]{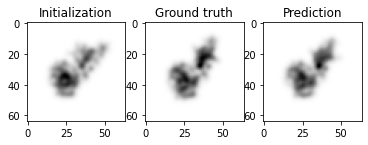}
\includegraphics[scale=.4]{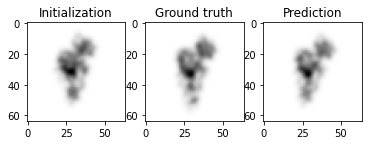}
\includegraphics[scale=.4]{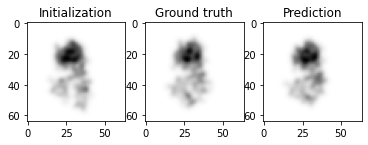}
\caption{Clean 2D tomographic projections of the reconstructed 3D structures. At the left we see the structure in the given conformation, in the middle the ground true and at the right the reconstruction.}
\end{subfigure}
\hfill
\begin{subfigure}[c]{.48\textwidth}
\centering
\includegraphics[scale=.3]{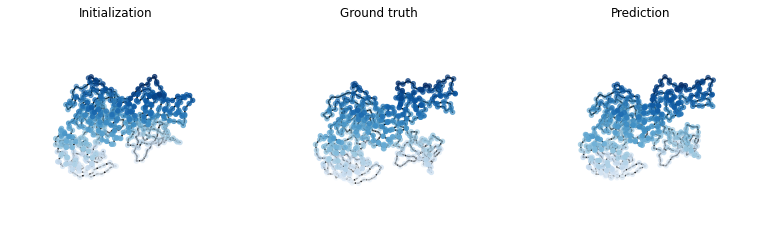}\\
\includegraphics[scale=.3]{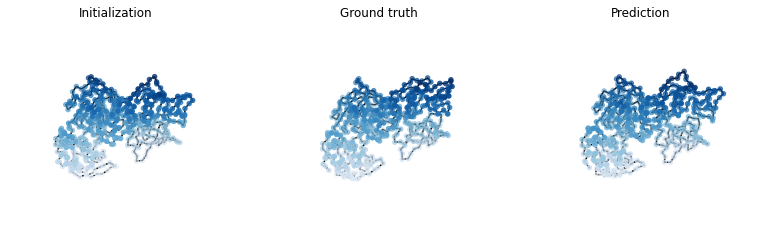}\\
\includegraphics[scale=.3]{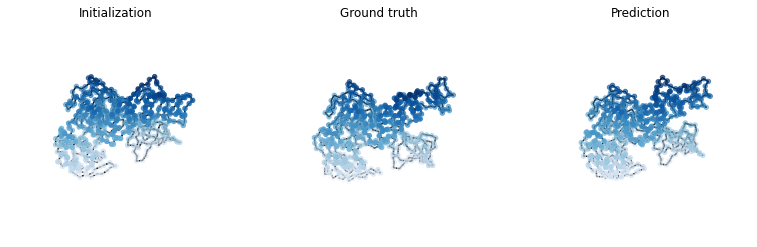}
\caption{At the left we see the structure in the known conformation, in the middle we see the particle in its specific conformation (ground truth) and at the right the 3D atomic structure predicted by means of our method.}
\end{subfigure}
\caption{Some examples of the tomographic reconstruction in subsection \ref{subsubsec: example 3d 2}. The reconstructions are taken from the experiment with noise levels as in 4-th row in Table \ref{talbe: expermients 2}.}
\label{fig: 3d reconstruction 2}
\end{figure}

\section{Conclusions and future steps}
\label{sec: conclusions and future}

In this work we present a method to recover the heterogeneous 3D structure of a flexible macromolecule from a cryo-EM dataset  and a given known conformation.
The method consists in estimating the deformation of the 3D structure in each cryo-EM image with respect to the given known conformation.
The core idea is to combine a known technique in manifold learning, used to approximate the manifold of conformations,  with a parametrization of the atomic structures satisfying a structural prior, in order to construct the function which maps the manifold of conformations to the corresponding atomic structure.

When combining these two ingredients, we use the following pieces of information, which can be obtained from existing techniques or databases:
\begin{enumerate}
\item An estimation of the pose of the particle in each cryo-EM image, which can be obtained from RELION \cite{scheres2012relion}, cryoSPARC \cite{punjani2017cryosparc} or FREALIGN \cite{lyumkis2013likelihood}.
\item A low-dimension representation of the conformations in each cryo-EM image. This can be obtained by means of a PCA analysis as in \cite{anden2018structural} or by using the representation of the images in the latent space of a variational autoencoder, as for instance in cryoDRGN \cite{zhong2021cryodrgn}.
\item We need to have access to the atomic model of the macromolecule in a specific conformation. This one can be obtained from the protein database PDB, or from an AlphaFold prediction of the 3D structure.
\item In order to reduce the number of parameters to be estimated in our approach, we can exploit further knowledge about the macromolecule such as secondary structures and other rigid substructures.
\end{enumerate}

On one hand, in most of the existing methods in single-particle cryo-EM,  the problem of recovering heterogeneous structures is addressed by estimating the 3D density of the macromolecule of interest in each conformation. On the other hand,  biological knowledge and new computational tools provide detailed information about the structure of the macromolecule, typically corresponding to one specific conformation.
In this work we attempt to close the gap between the both approaches by combining the information extracted from the cryo-EM dataset, represented by the first two elements in the above list, with the biological knowledge, represented by the other two elements.
As a result, we obtain an estimation of the atomic model for each of the conformations in the heterogeneous cryo-EM dataset. In the recent work \cite{rosenbaum2021inferring}, the same problem is addressed by means of a variational auto-encoder approach.

An important approach to analyse conformational deformations of flexible macromolecules is Normal Mode Analysis NMA \cite{doruker2000dynamics, tama2004flexible,  tama2004normal,  tama2001conformational, tirion1996large}. 
NMA is based on the harmonic dynamics of a potential energy function around a minimum energy conformation. The possible deformations of the atomic model are obtained in NMA from the principal eigenvectors, or normal modes,  of the Hessian matrix of the energy functional, and the corresponding frequencies. 
The space generated by the normal modes can therefore be interpreted as a local approximation of the manifold of conformations of the macromolecule around the  minimum energy conformation.
This might be in connection with the manifold learning technique used in our approach, in which the approximation of the manifold of conformation is obtained from the cryo-EM data.
However, the nature of the approximation is different in both cases. Whereas  the  manifold learning technique consists on a global approximation of the spectral properties of the  manifold of conformations, 
NMA gives a local approximation of the possible deformations, which can be seen as the tangent space of the manifold at the initial conformation.

Numerical experiments in synthetic data show the potential of our method.
However,  the accuracy of the predictions given by our method may rely on the availability and the accuracy of the information listed above. 
In the following, we list some future steps that may continue from this work, in order to make our method applicable to real cryo-EM datasets:
\begin{enumerate}
\item In our numerical experiments, the matrix of similarities between the paritcels is constructed by means of a Gaussian kernel, which gives rise to a dense matrix. Additionally, due to the reduced size of our datasets (4000 cryo-EM images), we can accurately compute the eigenvectors of the Laplacian matrix. 
However, real cryo-EM datasets are considerably larger ($\sim 10^{5-6}$ images), which makes it computationally infeasible. This can be overcome by considering a sparse matrix of similarities and by using methods to approximate the principal eigenvectors (such as randomized SVD \cite{halko2011finding} or Nystr{\"o}m method \cite{williams2000using}). 
The sensibility of our method to these approximations needs to be analysed.
\item In some cases, the pose estimation of the particle in a cryo-EM image may depend on the specific conformation.  This needs to be taken into account in the parametrization of the atomic structure that we use in our method, which assumes that the pose of the particle is independent of the conformation.  
This might be overcome by updating the pose estimations after the application of every gradient step in the optimisation algorithm.
\item The sensibility of our method to the low-dimension representation of the conformation needs to be further investigated.  In particular,  it would be interesting to compare the manifold of conformations obtained from PCA coordinates and the one obtained from the latent space representation in a VAE.
\item Finally, in this work we only treat the case of atomic structures forming a single chain. The method can of course be adapted to consider more complex  atomic models. For a protein, it should be possible to recover, not only the deformation of the backbone, but also the orientation of the side chains.
\end{enumerate}

\appendix

\section{Discrete Frenet Frames}
\label{appdx: atomic model param}

In this section, we  check that any atomic model (point cloud $\{ z_1, z_2, \ldots , z_m \} \in \R^{3m}$) satisfying the discrete curve property
\begin{equation}
\label{discrete curve prop}
\exists \delta >0 \quad \text{such that} \quad
\| z_{j+1}  - z_j \| = \delta, \qquad \forall j\in \{ 1, \ldots , m-1 \},
\end{equation}
where $\{ z_1, z_2, \ldots , z_m \}$ are the positions of the atoms in the atomic model,
can be represented by the parameters $(\Theta, \Psi, \hat{z} ,\hat{F})$ introduced in subsection \ref{subsec: structure spectral decomp}.
The construction of the discrete curve using this parameters is based on a discrete version of the Frenet Frames, which is extensively used in the description of smooth curves, and we use the same approach and notation as in \cite{hu2011discrete}.
This representation of the atomic structure, in which the parameters $(\hat{z}, \hat{F})$ represent the spatial location and the orientation of the structure and the parameters $(\Theta, \Psi)$ represent the shape of the discrete curve up to translations and rotations is crucial in our approach.
Although this representation of discrete curves by using the torsion and bond angles is standard and can be found in many works, we include it here for completeness.

As mentioned in subsection \ref{subsec: structure spectral decomp}, we need to select a reference atom $j_0\in \{1,\ldots , m-1\}$ (a reference point in the discrete curve), which will serve as initial condition to construct the rest of the point cloud.
This choice is arbitrary, but it is preferable to select an atom from a non-flexible part of the molecule (a part of the structure invariant over the different conformations). 

Let us define the map
\begin{equation}
\label{backbone param def}
\begin{array}{cccc}
\mathbf{Z} : & \mathcal{A} \times \mathcal{I} & \longrightarrow & \R^{3m} \\
\noalign{\vspace{2pt}}
& ( \Theta, \Psi, \hat{z}, \hat{F}) & \longmapsto &
\mathbf{Z} (\Theta, \Psi ,  \hat{z}, \hat{F}),
\end{array}
\end{equation}
where  $\mathbf{Z} (\Theta, \Psi ,  \hat{z}, \hat{F}) = \{ z_1, z_2, \ldots , z_m\} \in \R^{3m}$ is the solution to the dynamical system\footnote{The product $e_3 F_j$ denotes the matrix multiplication of the row vector $e_3 = (0,0,1)$ by the square matrix $F_j\in SO(3)$. Hence, $e_3F_j$ simply represents the third row of the matrix $F_j$.}
\begin{equation}
\label{dynamical system DFF}
\begin{cases}
z_{j+1} = z_j + \delta e_3 F_j & j\in \{ 1, \ldots , m-1\} \\
F_j = R (\theta_j, \psi_j) F_{j-1} & j\in \{ 2, \ldots , m-2\} \\
\text{with} \ z_{j_0}=\hat{z} \in \R^3 \ 
\text{and} \ F_{j_0} = \hat{F} \in SO(3),
\end{cases}
\end{equation}
with $\delta>0$ fixed, $e_3 = (0,0,1)$ and $R (\theta_j, \psi_j)\in SO(3)$ being the rotation matrix
\begin{equation}
\label{R_j}
R  (\theta_j, \psi_j) = 
\left(\begin{matrix}
\cos \psi_j \cos \theta_j & \cos \psi_j \sin \theta_j & - \sin \psi_j \\
-\sin \theta_j & \cos \theta_j & 0 \\
\sin \psi_j \cos \theta_j & \sin \psi_j \sin \theta_j & \cos \psi_j
\end{matrix}\right),
\end{equation}
with $\Theta = (\theta_2, \theta_3, \ldots , \theta_{m-1})$
and $\Psi = (\psi_2, \psi_3, \ldots , \psi_{m-1})$.
We recall the abbreviated notation $\mathcal{A} := [-\pi , \pi]^{m-2} \times [-\pi,\pi]^{m-2}$ and $ \mathcal{I}:= \R^3 \times SO(3)$ that we already used in subsection \ref{subsec: structure spectral decomp}.

It is not difficult to see that the image of $\mathbf{Z}$ is contained in the set of point clouds $\{ z_1, \ldots , z_m\}\in \R^{3m}$ satisfying \eqref{discrete curve prop}.
Next we prove that the map $\mathbf{Z}$ is indeed onto, i.e., any point cloud satisfying \eqref{discrete curve prop} can be represented as $\mathbf{Z}(\Theta, \Psi, \hat{z}, \hat{F})$.
Although this result is well-known, we include it here for completeness and also to show the construction and notation of the discrete Frenet frames that we use throughout the paper.

\begin{lemma}
\label{lem: discrete curve param}
Let $m\in \N$ and $\delta>0$, and let $\{ z_1, z_2, \ldots , z_m\}\in \R^{3m}$ be any point cloud satisfying \eqref{discrete curve prop}.
Then, there exist $(\Theta , \Psi , \hat{z}, \hat{F})\in \mathcal{A}\times \mathcal{I}$ such that
 $\mathbf{Z} (\Theta,\Psi,\hat{z}, \hat{F}) = \{ z_1, z_2, \ldots , z_m\}$.
\end{lemma}

\begin{proof}
We use the same construction of the discrete Frenet Frames as  in \cite{hu2011discrete}.
Let  $\{ z_1, \ldots , z_m\}\in \R^{3m}$ be a point cloud satisfying \eqref{discrete curve prop} for some given $\delta >0$.
We need to compute the Frenet Frames at every point $z_j$ for $j\in \{1, \ldots , m-1\}$.
First we compute the sequence of unitary vectors $\mathbf{T}:=\{t_j\}_{j=1}^{m-1}$  given by
$$
t_j =\dfrac{ z_{j+1} - z_j}{\delta} , \qquad
\forall j\in \{ 1, \ldots , m - 1  \}.
$$
These are the directions of each segment in the discrete curve.
Next, we compute the sequence of binormal vectors $\mathbf{B}:=\{b_j\}_{j=1}^{m-1}$, given by
\begin{equation}
\label{binormal vec def}
b_j =
\dfrac{t_{j-1} \times t_{j}}{| t_{j-1} \times t_j |}
\qquad \forall j\in \{2, \ldots , m-1\} \qquad \text{and} \qquad
b_1 = b_2,
\end{equation}
and the sequence of normal vectors $\mathbf{N}:= \{  n_j\}_{j=0}^{m-1}$ as
$$
n_j = b_j \times t_j,
\qquad
\forall j \in \{ 1 , \ldots , m-1\}.
$$
Finally we define the sequence of Frenet frames as $\mathbf{F}:= \{ F_j \}_{j=1}^{m-1}$, where each frame $F_j\in SO(3)$ is the $3\times 3$ matrix that has $n_j$, $b_j$ and $t_j$ as rows, i.e.
\begin{equation}
\label{Frenet Frame def}
F_j:=
\begin{bmatrix}
n_j \\ b_j \\ t_j
\end{bmatrix},
\qquad
\forall j \in \{ 1 , \ldots , m-1\}.
\end{equation}

Given the sequence of $m-1$ Frenet frames $\mathbf{F}:= \{ F_j \}_{j=1}^{m-1}$, we can compute the sequence of $m-2$ rotation matrices
$\mathbf{R}:= \{ R_j\}_{j=2}^{m-1}$ given by
\begin{equation}
\label{Transition matrix def}
R_j := F_{j+1} F_j^\ast \qquad 
\forall i \in \{ 2, \ldots  , m -2\}.
\end{equation}
One can readily prove that the sequence $\{z_1, z_2, \ldots , z_m\}$ satisfies
$$
\begin{cases}
z_{j+1} = z_j + \delta e_3 F_j & j\in \{ 1, \ldots , m-1\} \\
F_j = R_j F_{j-1} & j\in \{ 2, \ldots , m-2\}.
\end{cases}
$$
Using the same arguments as in \cite[Section 3.B]{hu2011discrete}, each rotation matrix $R_j\in \mathbf{R}$, constructed as in \eqref{Frenet Frame def}--\eqref{Transition matrix def},  is given by
$R_j = R(\theta_j , \psi_j)$ defined in \eqref{R_j},
where
$\theta_j\in [-\pi,\pi]$ is the torsion angle
$$
\cos \theta_j = b_{j+1}\cdot b_j,
$$
and  $\psi_j\in [0, \pi]$ is the bond angle, and satisfies
$$
\cos \psi_j = t_{j+1}\cdot t_j.
$$
Note that, in view of \eqref{binormal vec def},  we have $\theta_1 = 0$.

Hence, given any sequence of $m$ points $\{ z_1,  z_2, \ldots , z_m\} \in\R^{3m}$ satisfying \eqref{discrete curve prop}, we can compute the associated sequence of torsion and bond angles $[\Theta, \Psi]\in \mathcal{A}$, which along with the position $\hat{z}$ and the orientation $\hat{F}$ of the curve at the $j_0$-th atom yield $\mathbf{Z} (\Theta, \Psi , \hat{z}, \hat{F}) = \{ z_1, \ldots , z_m\}$.
\end{proof}

\section{Manifold spectral representation}
\label{appdx: manifold spectral representation}

In this section,  we describe the method to construct the graph Laplacian, from the cryo-EM dataset, which is used in subsection \ref{subsec: approx spectr decomp} to approximate the spectral properties of the unknown manifold of conformations $\mathcal{M}$. 
As outlined in \cite{moscovich2020cryo}, 
 the low-resolution reconstruction of the heterogeneous particles in the dataset obtained by the method in \cite{anden2018structural} can be used to construct a weighted graph.
Each vertex in the graph corresponds to a particle in the dataset, and the weights in the edges joining any two vertices are estimates of the affinity between the low-dimensional reconstructions of the 3D densities (i.e. the similarity between the underlying conformation of the particles).

Let us choose $N' < N$, where $N\times N$ is the resolution in pixels of the cryo-EM images and $N'\times N' \times N'$ is the resolution of the voxel respresentation of the 3D densities that we aim to reconstruct.  As indicated in \cite[Subsection 3.3]{moscovich2020cryo}, this method is limited to a low-resolution reconstruction, i.e. $N'\ll N$.
In \cite{anden2018structural}, the volume density associated to each particle $\hat{U}_i\in \R^{N'\times N'\times N'}$ is estimated as
$$
\hat{U}_i = \hat{\mu} + \hat{V}_q \beta_i,
$$
where $\hat{\mu} \in \R^{N'\times N'\times N'}$ is the estimated average density, $\hat{V}_q \in [\R^{N'\times N'\times N'}]^q$ is a tensor with the $q$ principal components of the estimated covariance matrix $\hat{\Sigma}$ (the so-called eigenvolumes)  and $\beta_i\in \R^q$ are the PCA coefficients associated to the $i$-th particle.
Here, the volumes $\hat{U}_i$ are reconstructed independently of the viewing direction of each particle, and then, the PCA components $\beta_i$ depend only on the conformation of the underlying particle. 
They can actually be used as a low-dimensional representation of the underlying conformation.
In view of the second part of Assumption \ref{assumption: discrete curve},  these PCA coefficients $\beta_i\in\R^q$ can be seen as a discrete approximation of a $d$-dimensional manifold $\mathcal{M}\subset \R^{q}$, with $d<q$.

We stress that any dimensionality reduction of the 3D structures of the particles might be used, instead of the PCA coefficients, for the construction of the graph.
We only need that the low-dimension representations of the volumes are invariant under rotations and translations.
For instance, the representation of the particles in the latent space obtained by a trained variational auto-encoder \cite{rosenbaum2021inferring,zhong2021cryodrgn} may work as well.

We now construct a weighted graph using the low-dimension representation of the conformation in each particle, denoted by $\beta_i\in \R^q$.
The weight $W_{ij}$ between any two vertices $i$ and $j$ in the graph must represent the similarity between the underlying conformations of the particles $i$ and $j$.
In the numerical experiments presented in section \ref{sec: numerical experiments}, we  used a Gaussian kernel weights of the form
$$
W_{ij} = \operatorname{exp} \left[\dfrac{-\| \beta_i - \beta_j\|^2}{2\sigma^2}\right], 
$$
for some $\sigma>0$ fixed.  
It is well-known that the computational complexity to compute the eigenvectors of a graph Laplacian can be significantly reduced when the associated matrix is sparse. This can be achieved by setting to $0$ all the weights $W_{ij}$ below a certain threshold.
In the numerical experiments in subsection \ref{sec: numerical experiments}, due to the rather small size of the synthetic datasets (only 4000 images), we could use a dense matrix of similarities.

Another possibility is to consider binary weights obtained by applying a symmetric $k$ Nearest Neighbours (KNN), 
$$
W_{ij} := 
\begin{cases}
1, & \text{if $\beta_i\in NN_k(\beta_j)$ or $\beta_j\in NN_k(\beta_i)$} \\
0 & \text{otherwise},
\end{cases}
$$
where $\beta_i\in NN_k(\beta_j)$ means that $\beta_i$ is one of the $k$ nearest neighbours of $\beta_j$.
This choice provides a sparse matrix of similarities, which may increase the computational efficiency when computing the eigenvectors of the graph Laplacian.

Once we have constructed the matrix of similarities between the conformations of the particles in the cryo-EM dataset, we need to define an associated Laplacian.
In our numerical experiments we have used the symmetric normalised graph Laplacian, defined as
$$
L:= D^{-1/2}(D-W)D^{-1/2},
$$
where $D$ is the degree matrix, i.e. the diagonal matrix with entries given by $D_{ii}=\sum_{j=1}^n W_{ij}$.
Finally, the eigenvectors \eqref{matrix eigenvectors} are obtained as the eigenvectors of $L$ associated to the $K$ smallest eigenvalues.

\bibliographystyle{abbrv}
\bibliography{mybibfile}

\begin{thebibliography}{10}

\bibitem{anden2015covariance}
J.~And{\'e}n, E.~Katsevich, and A.~Singer.
\newblock Covariance estimation using conjugate gradient for 3{D}
  classification in cryo-{EM}.
\newblock In {\em 2015 IEEE 12th International Symposium on Biomedical Imaging
  (ISBI)}, pages 200--204. IEEE, 2015.

\bibitem{anden2018structural}
J.~And{\'e}n and A.~Singer.
\newblock Structural variability from noisy tomographic projections.
\newblock {\em SIAM Journal on Imaging Sciences}, 11(2):1441--1492, 2018.

\bibitem{barnett2017rapid}
A.~Barnett, L.~Greengard, A.~Pataki, and M.~Spivak.
\newblock Rapid solution of the cryo-{EM} reconstruction problem by frequency
  marching.
\newblock {\em SIAM Journal on Imaging Sciences}, 10(3):1170--1195, 2017.

\bibitem{basu2000feasibility}
S.~Basu and Y.~Bresler.
\newblock Feasibility of tomography with unknown view angles.
\newblock {\em IEEE Transactions on Image Processing}, 9(6):1107--1122, 2000.

\bibitem{basu2000uniqueness}
S.~Basu and Y.~Bresler.
\newblock Uniqueness of tomography with unknown view angles.
\newblock {\em IEEE Transactions on Image Processing}, 9(6):1094--1106, 2000.

\bibitem{Beckstein2018}
O.~Beckstein, S.~L. Seyler, and A.~Kumar.
\newblock Simulated trajectory ensembles for the closed-to-open transition of
  adenylate kinase from {DIMS} {MD} and {FRODA}.
\newblock 10 2018.

\bibitem{berger1971spectre}
M.~Berger, P.~Gauduchon, and E.~Mazet.
\newblock Le spectre d'une vari{\'e}t{\'e} riemannienne.
\newblock {\em Le Spectre d’une Vari{\'e}t{\'e} Riemannienne}, pages
  141--241, 1971.

\bibitem{cheng2015primer}
Y.~Cheng, N.~Grigorieff, P.~A. Penczek, and T.~Walz.
\newblock A primer to single-particle cryo-electron microscopy.
\newblock {\em Cell}, 161(3):438--449, 2015.

\bibitem{doruker2000dynamics}
P.~Doruker, A.~R. Atilgan, and I.~Bahar.
\newblock Dynamics of proteins predicted by molecular dynamics simulations and
  analytical approaches: Application to $\alpha$-amylase inhibitor.
\newblock {\em Proteins: Structure, Function, and Bioinformatics},
  40(3):512--524, 2000.

\bibitem{frank2006three}
J.~Frank.
\newblock {\em Three-dimensional electron microscopy of macromolecular
  assemblies: visualization of biological molecules in their native state}.
\newblock Oxford university press, 2006.

\bibitem{grant2018cistem}
T.~Grant, A.~Rohou, and N.~Grigorieff.
\newblock cis{TEM}, user-friendly software for single-particle image
  processing.
\newblock {\em elife}, 7:e35383, 2018.

\bibitem{grebenkov2013geometrical}
D.~S. Grebenkov and B.-T. Nguyen.
\newblock Geometrical structure of {L}aplacian eigenfunctions.
\newblock {\em siam REVIEW}, 55(4):601--667, 2013.

\bibitem{halko2011finding}
N.~Halko, P.-G. Martinsson, and J.~A. Tropp.
\newblock Finding structure with randomness: Probabilistic algorithms for
  constructing approximate matrix decompositions.
\newblock {\em SIAM review}, 53(2):217--288, 2011.

\bibitem{hu2011discrete}
S.~Hu, M.~Lundgren, and A.~J. Niemi.
\newblock Discrete {F}renet frame, inflection point solitons, and curve
  visualization with applications to folded proteins.
\newblock {\em Physical Review E}, 83(6):061908, 2011.

\bibitem{jumper2021highly}
J.~Jumper, R.~Evans, A.~Pritzel, T.~Green, M.~Figurnov, O.~Ronneberger,
  K.~Tunyasuvunakool, R.~Bates, A.~{\v{Z}}{\'\i}dek, A.~Potapenko, et~al.
\newblock Highly accurate protein structure prediction with {A}lpha{F}old.
\newblock {\em Nature}, 596(7873):583--589, 2021.

\bibitem{katsevich2015covariance}
E.~Katsevich, A.~Katsevich, and A.~Singer.
\newblock Covariance matrix estimation for the cryo-{EM} heterogeneity problem.
\newblock {\em SIAM journal on imaging sciences}, 8(1):126--185, 2015.

\bibitem{kurlberg2021formal}
P.~Kurlberg and G.~Zickert.
\newblock Formal uniqueness in ewald sphere corrected single particle analysis.
\newblock {\em arXiv preprint arXiv:2104.05371}, 2021.

\bibitem{lee2016spectral}
A.~B. Lee and R.~Izbicki.
\newblock A spectral series approach to high-dimensional nonparametric
  regression.
\newblock {\em Electronic Journal of Statistics}, 10(1):423--463, 2016.

\bibitem{liao2010classification}
H.~Y. Liao and J.~Frank.
\newblock Classification by bootstrapping in single particle methods.
\newblock In {\em 2010 IEEE International Symposium on Biomedical Imaging: From
  Nano to Macro}, pages 169--172. IEEE, 2010.

\bibitem{lyumkis2013likelihood}
D.~Lyumkis, A.~F. Brilot, D.~L. Theobald, and N.~Grigorieff.
\newblock Likelihood-based classification of cryo-{EM} images using {FREALIGN}.
\newblock {\em Journal of structural biology}, 183(3):377--388, 2013.

\bibitem{milne2013cryo}
J.~L. Milne, M.~J. Borgnia, A.~Bartesaghi, E.~E. Tran, L.~A. Earl, D.~M.
  Schauder, J.~Lengyel, J.~Pierson, A.~Patwardhan, and S.~Subramaniam.
\newblock Cryo-electron microscopy--a primer for the non-microscopist.
\newblock {\em The FEBS journal}, 280(1):28--45, 2013.

\bibitem{moscovich2020cryo}
A.~Moscovich, A.~Halevi, J.~And{\'e}n, and A.~Singer.
\newblock Cryo-{EM} reconstruction of continuous heterogeneity by {L}aplacian
  spectral volumes.
\newblock {\em Inverse Problems}, 36(2):024003, 2020.

\bibitem{penczek2002variance}
P.~A. Penczek.
\newblock Variance in three-dimensional reconstructions from projections.
\newblock In {\em Proceedings IEEE International Symposium on Biomedical
  Imaging}, pages 749--752. IEEE, 2002.

\bibitem{penczek2011identifying}
P.~A. Penczek, M.~Kimmel, and C.~M. Spahn.
\newblock Identifying conformational states of macromolecules by eigen-analysis
  of resampled cryo-{EM} images.
\newblock {\em Structure}, 19(11):1582--1590, 2011.

\bibitem{penczek2006estimation}
P.~A. Penczek, C.~Yang, J.~Frank, and C.~M. Spahn.
\newblock Estimation of variance in single-particle reconstruction using the
  bootstrap technique.
\newblock In {\em Single-Particle Cryo-Electron Microscopy: The Path Toward
  Atomic Resolution: Selected Papers of Joachim Frank with Commentaries}, pages
  389--404. World Scientific, 2006.

\bibitem{punjani2017cryosparc}
A.~Punjani, J.~L. Rubinstein, D.~J. Fleet, and M.~A. Brubaker.
\newblock cryo{SPARC}: algorithms for rapid unsupervised cryo-{EM} structure
  determination.
\newblock {\em Nature methods}, 14(3):290--296, 2017.

\bibitem{rosasco2010learning}
L.~Rosasco, M.~Belkin, and E.~De~Vito.
\newblock On learning with integral operators.
\newblock {\em Journal of Machine Learning Research}, 11(2), 2010.

\bibitem{rosenbaum2021inferring}
D.~Rosenbaum, M.~Garnelo, M.~Zielinski, C.~Beattie, E.~Clancy, A.~Huber,
  P.~Kohli, A.~W. Senior, J.~Jumper, C.~Doersch, et~al.
\newblock Inferring a continuous distribution of atom coordinates from
  cryo-{EM} images using {VAE}s.
\newblock {\em arXiv preprint arXiv:2106.14108}, 2021.

\bibitem{scheres2012bayesian}
S.~H. Scheres.
\newblock A {B}ayesian view on cryo-{EM} structure determination.
\newblock {\em Journal of molecular biology}, 415(2):406--418, 2012.

\bibitem{scheres2012relion}
S.~H. Scheres.
\newblock {RELION}: implementation of a {B}ayesian approach to cryo-{EM}
  structure determination.
\newblock {\em Journal of structural biology}, 180(3):519--530, 2012.

\bibitem{shaw2020molecular}
D.~Shaw.
\newblock Molecular dynamics simulations related to sars-cov-2.
\newblock {\em DE Shaw Research Technical Data}, 2020.

\bibitem{sigworth1998maximum}
F.~J. Sigworth.
\newblock A maximum-likelihood approach to single-particle image refinement.
\newblock {\em Journal of structural biology}, 122(3):328--339, 1998.

\bibitem{tama2004flexible}
F.~Tama, O.~Miyashita, and C.~L. Brooks~III.
\newblock Flexible multi-scale fitting of atomic structures into low-resolution
  electron density maps with elastic network normal mode analysis.
\newblock {\em Journal of molecular biology}, 337(4):985--999, 2004.

\bibitem{tama2004normal}
F.~Tama, O.~Miyashita, and C.~L. Brooks~Iii.
\newblock Normal mode based flexible fitting of high-resolution structure into
  low-resolution experimental data from cryo-em.
\newblock {\em Journal of structural biology}, 147(3):315--326, 2004.

\bibitem{tama2001conformational}
F.~Tama and Y.-H. Sanejouand.
\newblock Conformational change of proteins arising from normal mode
  calculations.
\newblock {\em Protein engineering}, 14(1):1--6, 2001.

\bibitem{tirion1996large}
M.~M. Tirion.
\newblock Large amplitude elastic motions in proteins from a single-parameter,
  atomic analysis.
\newblock {\em Physical review letters}, 77(9):1905, 1996.

\bibitem{tunyasuvunakool2021highly}
K.~Tunyasuvunakool, J.~Adler, Z.~Wu, T.~Green, M.~Zielinski,
  A.~{\v{Z}}{\'\i}dek, A.~Bridgland, A.~Cowie, C.~Meyer, A.~Laydon, et~al.
\newblock Highly accurate protein structure prediction for the human proteome.
\newblock {\em Nature}, 596(7873):590--596, 2021.

\bibitem{van1987angular}
M.~Van~Heel.
\newblock Angular reconstitution: a posteriori assignment of projection
  directions for 3d reconstruction.
\newblock {\em Ultramicroscopy}, 21(2):111--123, 1987.

\bibitem{varadi2022alphafold}
M.~Varadi, S.~Anyango, M.~Deshpande, S.~Nair, C.~Natassia, G.~Yordanova,
  D.~Yuan, O.~Stroe, G.~Wood, A.~Laydon, et~al.
\newblock {A}lpha{F}old {P}rotein {S}tructure {D}atabase: massively expanding
  the structural coverage of protein-sequence space with high-accuracy models.
\newblock {\em Nucleic acids research}, 50(D1):D439--D444, 2022.

\bibitem{vinothkumar2016single}
K.~R. Vinothkumar and R.~Henderson.
\newblock Single particle electron cryomicroscopy: trends, issues and future
  perspective.
\newblock {\em Quarterly reviews of biophysics}, 49, 2016.

\bibitem{von2008consistency}
U.~Von~Luxburg, M.~Belkin, and O.~Bousquet.
\newblock Consistency of spectral clustering.
\newblock {\em The Annals of Statistics}, pages 555--586, 2008.

\bibitem{vulovic2013image}
M.~Vulovi{\'c}, R.~B. Ravelli, L.~J. van Vliet, A.~J. Koster, I.~Lazi{\'c},
  U.~L{\"u}cken, H.~Rullg{\aa}rd, O.~{\"O}ktem, and B.~Rieger.
\newblock Image formation modeling in cryo-electron microscopy.
\newblock {\em Journal of structural biology}, 183(1):19--32, 2013.

\bibitem{williams2000using}
C.~Williams and M.~Seeger.
\newblock Using the nystr{\"o}m method to speed up kernel machines.
\newblock {\em Advances in neural information processing systems}, 13, 2000.

\bibitem{zhong2021cryodrgn}
E.~D. Zhong, T.~Bepler, B.~Berger, and J.~H. Davis.
\newblock Cryo{DRGN}: reconstruction of heterogeneous cryo-{EM} structures
  using neural networks.
\newblock {\em Nature methods}, 18(2):176--185, 2021.

\end{thebibliography}

\end{document}